\documentclass[preprint,12pt]{elsarticle}

\usepackage{amssymb}
\usepackage{color}
\usepackage{bm}
\usepackage{comment}
\usepackage{amsmath,amssymb,amsfonts,amsthm}
\usepackage{mathrsfs}
\usepackage{bm}
\usepackage{soul}
\usepackage{array}
\usepackage{dcolumn}
\usepackage{physics}
\usepackage{hyperref}

\newtheorem{lemma}{Lemma}

\newtheorem{corollary}{Corollary}

\DeclareMathOperator{\argmin}{argmin}

\begin{document}

\begin{frontmatter}

\title{Probabilistic forecast of nonlinear dynamical systems with uncertainty quantification}

\author[inst1]{Mengyang Gu\footnote{Equal contributions between the first two authors. Corresponding authors: Mengyang Gu (\href{mailto:mengyang@pstat.ucsb.edu}{mengyang@pstat.ucsb.edu}) and Diana Qiu (\href{mailto:diana.qiu@yale.edu}{diana.qiu@yale.edu}). }}

\affiliation[inst1]{organization={Department of Statistics and Applied Probability},
            addressline={University of California}, 
            city={Santa Barbara},
            postcode={93106}, 
            state={California},
            country={USA}}

\author[inst1]{Yizi Lin}

\author[inst2]{Victor Chang Lee}

\author[inst2]{Diana Y. Qiu}

\affiliation[inst2]{organization={Department of Mechanical Engineering and Materials Sciences},
            addressline={Yale University}, 
            city={New Haven},
            postcode={06520}, 
            state={Connecticut},
            country={USA}}

\begin{abstract}
Data-driven modeling is useful for reconstructing nonlinear dynamical systems when the underlying process is unknown or too expensive to compute.  
Having reliable uncertainty assessment of the forecast enables tools to be deployed to predict new scenarios unobserved before. In this work, 
we first extend parallel partial Gaussian processes for predicting the vector-valued transition function that links the observations between the current and next time points, and quantify the uncertainty of predictions by posterior sampling. Second, we show the equivalence between the dynamic mode decomposition and the maximum likelihood estimator of the linear mapping matrix in the linear state space model. The connection provides a {probabilistic generative} model of dynamic mode decomposition and thus, uncertainty of predictions can be obtained. 
Furthermore, we draw close connections between different data-driven models for approximating nonlinear dynamics, 
through a unified view of {generative} models. 
We study two numerical examples, where the inputs of the dynamics are assumed to be known in the first example and the inputs are unknown in the second example. The examples indicate that uncertainty of forecast can be properly quantified, whereas model or input misspecification can degrade the accuracy of uncertainty quantification.  

\end{abstract}


\end{frontmatter}


\section{Introduction}
\label{sec:intro}

Dynamical systems are ubiquitously used for describing natural phenomena, such as passive motions driven by thermodynamics \cite{coffey2012langevin} and phase transition from flocking \cite{vicsek1995novel,toner1998flocks}, and social behaviors, such as epidemiological processes \cite{hethcote2000mathematics}. As mathematical models typically contain unknown parameters, observations are often used for calibrating the models and filtering the noises to estimate the latent state of the dynamical system. Kalman filter \cite{kalman1960new} and  Rauch–Tung–Striebel smoother \cite{rauch1965maximum}, for instance, produce fast estimation of latent states for linear dynamical systems with additive Gaussian fluctuations and noises, at a computational cost linearly increasing with the number of time points. When dynamical systems are nonlinear or non-Gaussian, approximate approaches, such as extended Kalman filter \cite{julier2004unscented}, 
particle filters \cite{kitagawa1996monte} and ensemble Kalman filter \cite{evensen1994sequential}, were developed to approximate the posterior distributions of latent states. However,  these approaches require underlying  {data-generating} models to be known, whereas models that exactly reproduce the reality may be unavailable or too costly to compute in some applications.

Data-driven approaches become useful for estimating dynamical systems when the true {data-generating} mechanism is unknown.  For instance,  orthogonal basis is estimated in proper orthogonal decomposition \cite{sirovich1987turbulence,berkooz1993proper} to reconstruct the covariance between each of the output coordinates by treating temporal observations as independent measurements. 
Dynamic mode decomposition  \cite{schmid2010dynamic,tu2014dmd} reconstructs the output vector through linearizing the one-step-ahead transition operator between the input and output pairs, where the eigenpairs of the linear mapping matrix produce a finite-dimensional approximation of the Koopman modes and eigenvalues   
\cite{rowley2009spectral,Brunton_koopman_siam_review_2022}. Extensive variants of Koopman operator have been proposed, such as utilizing longer temporal lag of observations through Hankel method or higher order dynamic mode decomposition \cite{le2017higher}, and utilizing nonlinear basis functions for lifting the process by the extended dynamic mode decomposition \cite{williams2015data}. A few recent techniques, such as sparse regression \cite{kaiser2018sparse,doi:10.1073/pnas.1517384113}, model predictive control \cite{korda2018linear}, and Koopman eigenfunctions \cite{folkestad2020extended}, were studied for designing the nonlinear basis and estimating the lifted state in extended dynamic mode decomposition. {Other nonlinear methods employ neural networks  to model differential equations \cite{chen2018neural,kovachki2023neural,li2021fourier}.}
 The uncertainty of the estimation by the dynamic mode decomposition and other machine learning approaches, however, may not be available, as the {generative} models were not well-studied.

Deploying data-driven models to forecast or extrapolate the input space requires reliable uncertainty assessments of predictions. 
We evaluate the precision of uncertainty quantification by the percentage of held-out observations covered by the  $1-\alpha$ predictive intervals, with $0<\alpha<1$. An efficient approach should have around $1-\alpha$ of the held-out samples covered by a short predictive interval. Gaussian process{es} have been applied in emulating expensive computer simulations \cite{sacks1989design,bayarri2007framework,li2022efficient},  and  calibrating computer models \cite{kennedy2001bayesian,zhao2022bayesian,chang2022ice}. 
However, uncertainty assessment of these approaches is typically assessed for interpolating physical input space, while {forecast or extrapolation} is required in various real-world applications, such as optimizing chemical reaction conditions through Bayesian optimization  \cite{shields2021bayesian} and controlling the predictive error by active learning \cite{fang2022reliable}.

The goal of this paper is to quantify the uncertainty from probabilistic forecasts by different approaches. Our contributions are three-fold. First, we extend a recent approach, called the parallel partial Gaussian process (PP-GP) \cite{gu2016parallel}, 
to forecast dynamical systems with multivariate outputs through predicting the one-step-ahead transition function, and the uncertainty of the forecast can be assessed through posterior sampling. 
 {Predicting the one-step-ahead transition function with a finite-dimensional space allows one to transform the forecast problem to an interpolation problem, and under regularity conditions, Gaussian process regression converges to the truth with a minimax rate  \cite{van2008rates}. Furthermore, the assessed uncertainty from the PP-GP can alert when the forecast becomes less accurate, which allows for timely interventions to control predictive error. 
The  PP-GP is particularly suitable for dynamical systems with massive outputs, as the computational complexity of the PP-GP is linear to the number of outputs at each time point.
}
 Second, we introduce a general two-step approach to derive a {probabilistic generative} model. Based on this approach, we show the estimation of dynamic mode decomposition is equivalent to the maximum likelihood estimator of a linear mapping matrix in a linear state space model, which produces uncertainty assessment from the sampling model. Third, we draw connections between different approaches, including Gaussian processes, proper orthogonal decomposition and dynamic mode decomposition. 
These connections allow one to examine the inherent {generative models} of different approaches, and to develop a suitable predictive model for real-world problems. 

We compare the approaches for forecasting and uncertainty quantification by two numerical examples. In the first example, we assume the inputs of the process are known, whereas we do not have prior knowledge of the functional form of the process. Hence we cannot use the exact form of the function to form the nonlinear basis. Rather we aim to test default or generic kernels or basis, which can be used for other scenarios. We test this scenario by the 
Lorenz 96 system \cite{lorenz1996predictability}, a benchmark approach of modeling atmospheric quantities at equally spaced locations along a cycle that induces chaotic behaviors.  The PP-GP  can detect the time when the predictive error becomes large, based on its internal uncertainty assessment of the forecast. 

In the second example, we do not assume the true inputs are known, which is unconventional in designing data-driven approaches, but not uncommon in practice \cite{gu2023data}.   
We consider one of the most challenging problems in condensed matter physics: simulating quantum many-body systems far from equilibrium. Many problems in quantum dynamics,  such as the motion of atoms or charge carriers, cannot be modeled by well-established equilibrium methods, including density functional theory (DFT) for the electronic ground state or the GW plus  Bethe–Salpeter equation \cite{Hybertsen1986, Rohlfing2000} which describes excited-state properties in the equilibrium linear response regime. This limits, for example, the understanding of systems under irradiation by ultrafast or intense laser pulses 
for obtaining information about the electronic structure of a system~\cite{sangalli2016nonequilibrium}. Therefore, nonequilibrium simulations that describe how a system responds to an external perturbation and how it evolves from one configuration to another under these circumstances are crucial for a complete understanding of electronic and optical properties of molecules and solids. 

A rigorous approach to simulating materials' nonequilibrium dynamics lies in propagating the nonequilibrium Green's function as a two-point correlator of the creation and annihilation field operators on the Keldysh contour~\cite{Keldysh1965,KB1962, Stefanucci2013}. This approach has been recently applied to compute various nonlinear and nonequilibrium optical responses from first principles in the adiabatic limit, which limits the time-evolution to a single average time, neglecting memory effects~\cite{Attaccalite2011,chan2021giantexciton, Perfetto2016,perfetto2019nonequilibrium, Perfetto2020}. However, even in the adiabatic approximation, the numerical evaluation is far from trivial, requiring millions of CPU hours for systems of only a few atoms. Thus, models that can forecast the time-evolution without the need of the simulator are urgently needed. 
Recent work \cite{yin2022using,yin2023analyzing} uses dynamic mode decomposition to approximate the Green's functions where the system is assumed to start from a known non-interacting state, 
and it is driven by an arbitrary external electromagnetic field. 
Different representations of the Green's function encode spectroscopic information, which may be measured in experiment{s}.  In this work, inputs such as the external field and many-electron interactions are not used in constructing data-driven models to test uncertainty quantification of forecast for the scenario when inputs are misspecified.

The article is organized {as follows}. In Section \ref{sec:prob_forecast}, we extended the PP-GP for forecasting nonlinear dynamical processes. We introduce {a general approach to derive a {generative} model, and apply this approach to derive a sampling model of the dynamic mode decomposition and its predictive distribution in Section \ref{sec:dmd_variants}}.  PP-GP, dynamic mode decomposition and proper orthogonal decomposition are compared in Section \ref{sec:connection}, focusing on {generative} models of these approaches. In Section \ref{sec:numeric}, we numerically compare different approaches for forecast, and discuss the scenarios when reliable forecast can be constructed even at reasonably long trajectory.
We conclude this study and outline future directions in Section \ref{sec:conclusion}. The data and code used in this work are publicly available (\url{https://github.com/UncertaintyQuantification/forecast_dynamical_systems}). 

\section{Probabilistic forecast and uncertainty quantification through parallel partial Gaussian processes}
\label{sec:prob_forecast}

The parallel partial Gaussian process (PP-GP) emulator was originally designed as a fast surrogate model to approximate computationally expensive computer models with massive observations \cite{gu2016parallel}. Emulating computer models typically starts with running the computer simulation at a set of `space-filling' designs, such as the Latin hypercube designs \cite{sacks1989design}, for building the emulator. For any other inputs untested before, the predictive distribution of the emulator is used for predictions and quantifying the uncertainty of the predictions. Most of the computer model emulation tasks deal with \textit{interpolation} for a design space, meaning that the distance between the test input and some training inputs is close, as the `space-filling' inputs fill the input space. However, many scientific tasks, such as designing a new molecule or forecasting dynamical systems, inevitably require  \textit{extrapolation} from the existing design space, where reliable uncertainty quantification of the predictions is needed. Here we extend the PP-GP model for forecasting nonlinear dynamical systems that enables uncertainty of the forecast to be quantified in a probabilistic way, which was not studied before.  

Suppose we have collected $n$ vectors of real-valued outputs or snapshots, each having $m$ dimensions, where the $t$th output vector is denoted as $\mathbf y(\mathbf x_t)=(y_1(\mathbf x_t),..., y_m(\mathbf x_t))^T$, with $\mathbf x_t$ being a $p$-dimensional input that contains the observations in the prior time points and additional physics input, for $t=1,...,n$.  When the observational vector in the prior time point is used as the input, we have $\mathbf x_t=\mathbf y(\mathbf x_{t-1})$ and thus $p=m$. In general, the dimensions of the input and output vectors can be different. 

In the PP-GP model, we assume a distinct mean parameter $\mu_j$ and variance parameter $\sigma^2_j$ at each coordinate of the output $y_j(\cdot)$, for $j=1,...,m$, which makes it flexible to capture the scale difference in output coordinates. The correlation of the output at two coordinates {$j$ and $j'$}, on the other hand, is assumed to be the same, i.e. $\mbox{Cor}( y_j(\mathbf x),  y_j(\mathbf x'))=\mbox{Cor}( y_{j'}(\mathbf x),  y_{j'}(\mathbf x'))$, which greatly simplifies the computation.
The observations from dynamical system may contain noises, due to numerical or measurement errors. Thus, for any input $\mathbf x$,  we define the PP-GP model of the $j$th coordinate below: 
\begin{equation}
y_j(\mathbf x)=f_j(\mathbf x)+ \epsilon_{j},
\end{equation}
where $\mathbf x$ is p-dimensional input variable, $f_j(\cdot)$ follows a Gaussian process prior with mean $\mu_j$ and covariance $\sigma^2_jK(\cdot,\cdot)$, and $\epsilon_{j}$ is an independent Gaussian noise with variance $\eta \sigma^2_j$.  For any $p\times n$ input matrix $\mathbf X=[\mathbf x_1,...,\mathbf x_n]$, integrating the latent Gaussian process $f_j(\cdot)$, the marginal distribution $\mathbf y_j=[y_j(\mathbf x_1),...,y_j(\mathbf x_n)]^T$ follows a multivariate normal distribution: 
\begin{equation}
(\mathbf y_j \mid \mathbf X, \mu_j, \sigma^2_j, \eta, \bm \gamma )  \sim \mathcal{MN}\left( \mu_j \mathbf 1_n, \sigma^2_j \tilde {\mathbf K}\right),
\end{equation}
where $\mathcal{MN}$ denotes the multivariate normal distribution,  $\mu_j$ is an unknown mean parameter,   $\tilde {\mathbf K}=(\mathbf K +\eta\mathbf I_n)$ with $\mathbf K$ being a correlation matrix and  {$\eta$ being a nugget parameter  due to noises in observations}. The $(t,t')$th entry of $\mathbf K$ follows  $K_{t,t'}=K(\mathbf x_t,\mathbf x_{t'})$ with $K(\cdot,\cdot)$ being a kernel function containing a {$\tilde p$-dimensional} 
  range parameter $\bm \gamma$,  and $\mathbf I_n$ denotes an identity matrix. {With a nugget parameter, the prediction of GP will be dragged towards to the mean compared to a noise-free GP \cite{andrianakis2012effect}.}  Additional trend or mean basis functions can be included in the mean of the PP-GP model. 

Frequently used covariance functions include isotropic covariance and product covariance \cite{sacks1989design}. The isotropic covariance is a function of Euclidean distance between two inputs:  $d=||\mathbf x-\mathbf x'||$ with $||\cdot||$ denoting the $L_2$ norm. For instance, the isotropic power exponential covariance function follows 
\begin{equation}
\sigma^2_j K(d)= \sigma^2_j \exp\left(-\frac{d^{\alpha}}{\gamma}\right),
\end{equation}
where the roughness parameter $0<\alpha\leq 2$ is typically held fixed and $\gamma$ is a positive range parameter controlling the correlation length, which is often estimated by data. Here {the number of range parameters is 1}, i.e. $\tilde p=1$.  

Another widely used isotropic covariance is the Mat{\'e}rn covariance \cite{handcock1993bayesian}:  
 \begin{equation}
    \sigma^2_j K(d)= \sigma^2_j \frac{2^{1-\alpha}}{\Gamma(\alpha)}\left(\frac{\sqrt{2\alpha}d}{\gamma} \right)^{\alpha} \mathcal{K}_\alpha\left(\frac{\sqrt{2\alpha} d}{\gamma} \right) , 
 \end{equation}
where $d=||\mathbf x-\mathbf x'||$ is the distance between inputs, $\Gamma(\cdot)$ is the gamma function and $\mathcal{K}_{\alpha}(\cdot)$ is the modified Bessel function of the second kind with a positive parameter $\alpha$. The Mat{\' e}rn covariance has a closed-form expression with an integer roughness parameter $\alpha=\frac{2 z+1}{2}$ with $z \in \mathbb{N}$. When $\alpha=2.5$, for example, the Mat{\' e}rn covariance function has the following expression 
\begin{equation}
\begin{aligned}
  \sigma^{2}_j K(d) = \sigma^{2}_j \left(1+\frac{\sqrt{5} d}{\gamma}+ \frac{5 d^2}{3 \gamma^2}\right) \exp \left(-\frac{\sqrt{5} d}{\gamma}\right).
    \label{eq:matern_5_2_kernel}
\end{aligned}
\end{equation} 
The GP model having a Mat{\'e}rn covariance with a roughness parameter $\alpha$ is $\lfloor\alpha-1 \rfloor$ mean squared differentiable, an appealing property as the smoothness of the process is directly controlled by the roughness parameter $\alpha$. 

When the input variables have different scales, a product correlation function is more frequently used, as it allows one to have a distinct correlation length parameter for each of the input coordinates: 
\begin{equation}
 K(\mathbf x, \mathbf x')=\prod^p_{l=1} K_l(x_l,x'_l ),
\label{equ:product_kernel}
\end{equation}
where $K_l(x_l,x'_l )$ is a kernel function,  such as power exponential kernel or Mat{\'e}rn kernel, with range parameter $\gamma_l$ for $l=1,...,p$ and we have $\tilde p=p$ range parameters. The product form of the kernel in Eq.~(\ref{equ:product_kernel}) is widely used for computer model emulation \cite{Bayarri09,conti2010bayesian,anderson2019magma} and often treated as the default setting in statistical emulator software packages \cite{roustant2012dicekriging,gu2018robustgasp}, as inputs variable can contain completely different scales and physical meanings, thus requiring distinct correlation lengthscales.  In practice, isotropic covariance may be used when  Euclidean distance is meaningful for characterizing the distance between two inputs. The product covariance may yield better predictive performance, whereas more range parameters are needed to be estimated.

In the PP-GP model, we have $m$ mean parameters $\bm \mu=(\mu_1,...,\mu_m)^T $, $m$ variance parameters $\bm \sigma^2=(\sigma^2_1,...,\sigma^2_m)^T$, $\tilde p$ covariance range parameters $\bm \gamma=(\gamma_1,...,\gamma_{\tilde p})^T$ and a nugget parameter $\eta$. We follow the Bayesian procedure to define the prior of parameters. As most of the parameters can be integrated out explicitly, meaning that the uncertainty from the estimates of these parameters is quantified during the Bayesian inference with almost no extra computational cost,  whereas a plug-in estimator of the parameters may ignore the uncertainty from parameter estimation. 
We assume an objective Bayesian prior \cite{gu2016parallel,berger2001objective} of the parameters below 
\begin{equation}
 \pi(\bm \mu,\bm \sigma^2, \bm \gamma, \eta ) \propto \frac{\pi(\bm \gamma, \eta)}{\prod^m_{j=1}\sigma^2_{j}}, 
 \label{equ:ref_prior_mean_var}
 \end{equation}
where $\pi(\bm \gamma, \eta)$ is a prior of the range and nugget parameters. Denote the $m\times n$  matrix of observations  by $\mathbf Y=[\mathbf y(\mathbf x_1),...,\mathbf y(\mathbf x_n)]$. Integrating out the mean and variance parameters, the predictive distribution for  $\mathbf x_{t^*}$ follows a noncentral Students' t distribution with $n-1$ degrees of freedom \cite{gu2016parallel}: 
\begin{equation}
    \left(y_j\left(\mathbf x_{t^*}\right) \mid \mathbf X, \mathbf Y, \mathbf x_{t^*}, \bm \gamma, \eta\right) \sim \mathcal T\left(\hat{y}_j\left(\mathbf x_{t^*}\right), \hat{\sigma}_j^2 K^{*}, n-1\right),
    \label{equ:pred_distr}
\end{equation}
where the predictive mean and scale parameters follow 
\begin{align}
\hat{y}_j\left(\mathbf x_{t^*}\right)&=  \hat{\mu}_j+\mathbf{k}^T(\mathbf x_{t^*}) \mathbf {\tilde K}^{-1}\left(\mathbf y_j-\hat{\mu}_j \mathbf{1}_n\right), \label{equ:hat_y_j}\\
\hat{\sigma}_j^2&=  \frac{1}{n-1}\left(\mathbf{y}_j- \hat{\mu}_j\mathbf{1}_n\right)^T {\mathbf{\tilde K}}^{-1}\left(\mathbf{y}_j- \hat{\mu}_j\mathbf{1}_n\right), \label{equ:sigma_j_2} \\
 K^{*}&= 1+\eta -\mathbf{k}^T(\mathbf x_{t^*}) {\mathbf{\tilde K}}^{-1} \mathbf {k}(\mathbf x_{t^*}) +\frac{\left(1-\mathbf{1}_n^T {\mathbf{\tilde K}}^{-1} \mathbf{k}(\mathbf x_{t^*})\right)^2}{\mathbf{1}_n^T {\mathbf{\tilde K}}^{-1} \mathbf{1}_n},
\end{align} 
with ${\mathbf{\tilde K}}=\mathbf K +\eta \mathbf I_n$,  $\hat \mu_j=\left(\mathbf{1}_n^T {\mathbf{\tilde K}}^{-1} \mathbf{1}_n\right)^{-1}\mathbf{1}_n^T {\mathbf{\tilde K}}^{-1} \mathbf y_j$ being the generalized least square estimator of the mean, $\mathbf 1_n$ is an n-vector of ones and $\mathbf k(\mathbf x_{t^*})=(K(\mathbf x_1,\mathbf x_{t^*}),..., K(\mathbf x_n,\mathbf x_{t^*}))^T$ being an $n$-vector of the covariance between the training inputs and the test input. 

The PP-GP model has been implemented in different computational platforms such as MATLAB,  Python and R \cite{gu2018robustgasp}.   
 The predictive mean from distribution in Eq.~(\ref{equ:pred_distr}) is typically used for prediction. The uncertainty of the prediction can be quantified by the predictive interval. 
 {A few recent studies approximate the transition operator of dynamical systems through kernel flows  \cite{hamzi2021simple, hamzi2021learning}}. In comparison, the PP-GP provides a distinct mean and variance for each coordinate, and these parameters are integrated out for calculating the predictive distribution in Eq. (\ref{equ:pred_distr}), making the model more flexible in uncertainty assessment.  We will discuss the computational issue and compare PP-GP with other vector-valued GP approaches in Section \ref{subsec:pp-gp-computation}.

\subsection{Predictions as weighted averages of basis functions and output vectors}
The predictive mean or median $\hat{y}_j(\mathbf x_{t^*})$ is often used for one-step-ahead prediction of output coordinate $j$ for any test input $\mathbf x_{t^*}$, for $j=1,...,m$. The corollary below shows that the prediction of the PP-GP model can be written as a weighted average of the observations and kernel functions. 

\begin{corollary}
The predictive mean vector $\mathbf{\hat y}(\mathbf x_{t^*})=({\hat y}_1(\mathbf x_{t^*}),...,{\hat y}_m(\mathbf x_{t^*}))^T$ from  Eq.~(\ref{equ:hat_y_j}) follows 

\begin{equation}
 \mathbf{\hat y}(\mathbf x_{t^*})= \mathbf Y \mathbf{ v}^T= \bm{\hat \mu}+ \mathbf { W}  \mathbf k(\mathbf x_{t^*}),
 \label{equ:weighted_average_pred_mean}
 \end{equation} 
where $\mathbf v $ is an $n$-dimensional row vector and $\mathbf{ W}=[\mathbf{ w}^T_1,...,\mathbf{ w}^T_m]^T$ is a $m \times n$ matrix with $\mathbf w_j$ being an $n$-dimensional row vector defined below: 
\begin{align}
     \mathbf{ v} &=\frac{(1-\mathbf k^T(\mathbf x_{t^*}) \mathbf{\tilde K}^{-1}\mathbf 1_n) \mathbf{1}_n^T {\mathbf{\tilde K}}^{-1}}{\left(\mathbf{1}_n^T {\mathbf{\tilde K}}^{-1} \mathbf{1}_n\right)}+\mathbf k^T(\mathbf x_{t^*}) \mathbf{\tilde K}^{-1}, \\
     \mathbf{ w}_j &= \left(\mathbf{y}_j- \hat{\mu}_j\mathbf{1}_n\right)^T \mathbf{\tilde K}^{-1}, 
\end{align}
 for $j=1,...,m$. 
\label{corollary:pred_mean}
\end{corollary}
From Corollary  \ref{corollary:pred_mean}, the prediction of a test input at the $j$th coordinate of the output from the PP-GP model can be written as a weighted average of the observations at the $j$th coordinate, $\hat{y}_j(\mathbf x_{t^*})=\mathbf{v} \mathbf y_j$. Furthermore, the residuals can be written as a weighted average of the kernel function between the test input and training input set, $\hat y_j(\mathbf x_{t^*})-\hat \mu_j=\mathbf w_j \mathbf k(\mathbf x_{t^*})$, as outlined by the second equality in Eq.~(\ref{equ:weighted_average_pred_mean}). 
When $\hat \mu_j=0$, the predictive mean estimator in Eq.~(\ref{equ:hat_y_j}) at each output coordinate is equivalent to the  kernel ridge regression separately for each coordinate $j$ \cite{gu2022theoretical}: %
\begin{equation}
\hat y_j(\cdot)=\mbox{argmin}_{f_j \in \mathcal H}\left\{\frac{1}{n}\sum^{n}_{t=1}(y_j(\mathbf x_t)-f_j(\mathbf x_t) )^2+\frac{\eta}{n}||f_j||^2_{\mathcal H} \right\},
\label{equ:krr}
\end{equation} 
where $\mathcal H$ is the reproducing kernel Hilbert space (RKHS) \cite{wendland2004scattered} attached to the kernel $K(\cdot,\cdot)$ and $||\cdot||_{\mathcal H}$ is the associated native norm. 
The loss function in {Eq.~(\ref{equ:krr})}  penalizes the fitting error and complexity of the model simultaneously, which helps avoid the overfitting problem automatically. Compared to the kernel ridge regression in Eq.~(\ref{equ:krr}),  the uncertainty of the prediction of PP-GP can be quantified based on the predictive distribution in  Eq.~(\ref{equ:pred_distr}).

Here the range and nugget parameters $(\bm \gamma, \eta)$ can be estimated by maximum marginal posterior distribution described in the Appendix and then these parameters will be plugged into Eq.~(\ref{equ:pred_distr}) for computing the predictive distribution. Here the uncertainty of the mean and variance parameters are taken into account in the analysis, whereas the uncertainty in estimating the range and nugget parameters was not considered due to computational feasibility, and confounding issues between kernel parameters  \cite{zhang2004inconsistent}. 
Sampling the parameters from the posterior distribution 
or residual bootstrap approach \cite{davison1997bootstrap} can be used for estimating the uncertainty of these kernel parameters.

\subsection{Forecast by parallel partial Gaussian processes}
 Here we focus on estimating the one-step-ahead transition function that maps the input $\mathbf x_t$ to the output $\mathbf y(\mathbf x_t)$ at any time point $t$.  
 Consider a simple scenario, where the previous snapshot is used for predicting the output at the current time step:  $\mathbf x_{t}=\mathbf y(\mathbf x_{t-1})$ for any $t\geq 1$. 
 Since the function is nonlinear, we may iteratively use the predictive distribution to sample $S$ chains for forecasting from $t^*=n+1,...,n+n^*$. 
 {Let} the input {be} $\mathbf x^{(s)}_{n+1}=\mathbf y (\mathbf x_n)$ for any chain $s$, $s=1,...,S$.  For each of the chain $s$,  we simulate a new output from the predictive distribution sequentially for $t^*=n+1,...,n+n^*$: 
\begin{equation}
\mathbf y^{(s)}(\mathbf x^{(s)}_{t^*+1})\sim p(\mathbf y^{(s)}(\mathbf x^{(s)}_{t^*+1}) \mid \mathbf Y, \mathbf X, \mathbf x^{(s)}_{t^*+1}, \bm \gamma, \eta),  
\end{equation}
where $p(\mathbf y^{(s)}(\mathbf x_{t^*+1}^{(s)}) \mid \mathbf Y, \mathbf X, \mathbf x^{(s)}_{t^*+1}, \bm \gamma, \eta)$ is the predictive distribution of $\mathbf y^{(s)}(\mathbf x_{t^*+1}) $.

As directly sampling from the joint predictive distribution at all output coordinates can be computationally intensive, one may sample the $j$th coordinate of the output from the marginal predictive distribution $p(y^{(s)}_{j}(\mathbf x_{t^*+1}) \mid \mathbf Y, \mathbf X, \mathbf x^{(s)}_{t^*+1}, \bm \gamma, \eta)$ in Eq.~(\ref{equ:pred_distr}) as an approximation. 
After we obtain predictive samples $y^{(s)}_{t^*,j}$ for $s=1,..., S$, we can use mean or median for predictions, and the lower and upper $\alpha$ quantiles of the samples for constructing the $1-\alpha$ predictive interval for any $0<\alpha<1$. Furthermore, the predictive mean $\hat{\mathbf y} (\mathbf x_{t^*+1})$ may be approximated by using a plug-in estimator of the input $ \mathbf x_{t^*+1} \approx \mathbf{\hat y}(\mathbf x_{t^*})$ by Eq.~(\ref{equ:hat_y_j}).  {The assessed uncertainty from PP-GP can alert when the forecast becomes
less accurate, which allows for timely interventions to control the predictive error, whereas the uncertainty assessment may not be not available in some other machine learning approaches \cite{kaiser2018sparse,chen2018neural,kovachki2023neural}. }

{Here we transform the forecast problem to predict the one-step-ahead transition function with a finite-dimensional input space. Under regularity conditions, the minimax convergence rate of Gaussian process regression to the truth was studied \cite{van2008rates}. We should note that the convergence is {obtained} when the training inputs fill {a bounded input} domain, whereas sequentially sampled input in forecast could move outside the training input domain in practice. When the inputs approach the boundary of input space, the quantified uncertainty becomes large, which can give alert for refining the prediction by expanding the input domain and training the PP-GP emulator with inputs in the expanded input domain. It is of interest to study the convergence of PP-GP forecast in different dynamical systems.}

\subsection{Computational complexity}
\label{subsec:pp-gp-computation}
{Compared with other GP approaches for emulating vectorized output}, one advantage of the PP-GP model comes with the computational scalability when the number of output coordinates $m$ is large. Computing the predictive mean of $m$ output coordinates in Eq.~(\ref{equ:hat_y_j})  requires $\mathcal O(nm)+\mathcal O(n^3)$ operations,  and to obtain $S$ predictive samples of $m$ output vectors at $n^*$ time points for uncertainty quantification requires $\mathcal O(n^2n^*S)+\mathcal O(n^2m)$ operations. The largest cost of PP-GP typically comes from estimating the range and nugget parameters, which requires $\mathcal O(\tilde Sn^3+\tilde Sn^2m)$ operations for $\tilde S$ iterations in numerical optimization. When the number of time points $n$ is large, approximation methods, such as the inducing point method \cite{snelson2006sparse} and the Vecchia approach \cite{vecchia1988estimation}, may be used for approximating the likelihood function of Gaussian processes. 

The computational advantage of PP-GP comes from two assumptions. First, the outputs at different coordinates are assumed to be independent. 
In Theorem 1 in \cite{gu2016parallel}, the authors show that the predictive mean of PP-GP is exactly the same as the predictive mean of a separable Gaussian process of the vectorized output, with the covariance $\bm \Sigma_y \otimes \mathbf K$, where $\bm \Sigma_y$ is the covariance of output at different coordinates, and the variance between the two models is similar.  The inverse of covariance between output coordinates $\bm \Sigma_y$ generally takes $\mathcal O(m^3)$ operations in computing the likelihood function of separable Gaussian process, whereas the complexity of predictions by PP-GP is linear to the number of coordinates ($m$). Second, the correlation of the output at different inputs $\mathbf x$ is shared across output coordinates.  Allowing the correlation parameters to differ at each spatial coordinate makes the computational complexity become $\mathcal O(n^3m)$ for computing the predictive mean, which is higher than $\mathcal O(nm)+\mathcal O(n^3)$. Furthermore, separably estimating $m(\tilde p+1)$ range and nugget parameters can be less stable as PP-GP.

\section{Generative models of dynamic mode decomposition}
\label{sec:dmd_variants}
 {Mathematical and machine learning approaches often minimize a loss function for estimating parameters.   Many {loss-minimization} approaches can be shown to be equivalent to statistical {estimations} {from} probabilistic {generative} models. 
We summarize a two-step procedure of building a {generative} model. 
        \begin{enumerate}
        \item Build a probabilistic {generative} model of untransformed data with parameters that can be identified from data. Generalize the model as much as possible to the extent that the parameters can still be identified from the data. 
        \item Show equivalence between the estimator that minimizes a loss function and a statistical estimator, such as the maximum likelihood estimator or maximum marginal likelihood estimator, of the probabilistic {generative} model. 
        \end{enumerate}
    {The generative model helps us understand the underlying assumptions from the loss-minimization estimator. A more efficient estimator can be derived based on the generative model if the loss-minimization estimator is not optimal}. We follow this procedure to derive a {generative} model of the dynamic mode decomposition that enables uncertainty assessment of the estimation. 
} 
\subsection{Dynamic mode decomposition}
Dynamic mode decomposition (DMD) is a data-driven approach to obtain a reduced rank representation of data from complex dynamical systems \cite{schmid2010dynamic}, which quickly gains popularity for approximating dynamical systems  \cite{schmid2022dynamic}.  
Here we summarize DMD  and derive a generative model of DMD.  

Let us split  the $m \times n$ real-valued observational matrix at $n$ time points  $\mathbf{Y}$  into two matrices, $\mathbf Y_{1:(n-1)}=[\mathbf{y}(\mathbf x_1),\mathbf{y}(\mathbf x_{2}),\dots, \mathbf{y}(\mathbf x_{n-1})]$ and $\mathbf Y_{2:n}=[\mathbf{y}(\mathbf x_2),\mathbf{y}(\mathbf x_{3}),\dots, \mathbf{y}(\mathbf x_{n})]$. 
 DMD  relies on the approximation:  $\mathbf{y}(\mathbf x_{t+1}) \approx \mathbf{A}\mathbf{y}({\mathbf x_t})$ for $t=1,\dots, n-1$, where $\mathbf{A}$ is an $m \times m$ matrix.  
In DMD,  $\mathbf A$ is estimated by minimizing the loss between the observations and the linear dynamics constructed from  previous time steps:
\begin{align}
\label{equ:objective_function_DMD}
    \hat{\mathbf A} = \argmin_{\mathbf A}\lVert \mathbf{Y}_{2:n}-\mathbf{A}\mathbf Y_{1:n-1} \rVert
    {= \mathbf{Y}_{2:n} (\mathbf{Y}_{1:{n-1}})^+,} 
\end{align}
where $\lVert \cdot \rVert$ is the $L_2$ norm or Frobenius norm {and $(\mathbf{Y}_{1:{n-1}})^+$ is the Moore–Penrose pseudo-inverse of $\mathbf{Y}_{1:{n-1}}$.}

We first introduce the lemma that connects the DMD estimation to the maximum likelihood estimator (MLE) of the linear mapping matrix in a dynamic linear model \cite{West1997} or linear state space model \cite{durbin2012time}. 

\begin{lemma}[Equivalence Between the MLE of the Linear  Mapping Matrix in a Linear State Space Model and the DMD Estimation]
 The DMD estimator  $\hat{\mathbf A}$ in Eq.~(\ref{equ:objective_function_DMD}) is the MLE of $\mathbf A$ of the following linear state space model 
 \begin{align} 
    \mathbf{y}(\mathbf x_{t+1}) = \mathbf{A} \mathbf y (\mathbf{x}_t) +  \bm{\varepsilon}_{t+1},
    \label{equ:dmd_process}
\end{align}
where {$ \bm{\varepsilon}_{t+1}\sim \mathcal{MN}(\mathbf 0, \bm \Sigma_{\varepsilon} )$ is a vector of Gaussian distributions with a positive definite covariance matrix $\bm \Sigma_{\varepsilon}$}, for any $t=1,2,...n$ and we assume the marginal distribution of initial state $\mathbf y(\mathbf x_1)$ does not depend on $\mathbf A$. 
{We refer the linear state model by Eq.~(\ref{equ:dmd_process})  the  DMD-induced process. }
 \label{lemma:equivalence_DMD_MLE}
 \end{lemma}

 \begin{proof}
 {
 The log-likelihood of $\mathbf A$ and $\bm \Sigma_\varepsilon$ is:
\begin{align*}
    \mathcal L(\mathbf A, \mathbf \Sigma_\varepsilon) &= \log \biggl\{p(\mathbf y(\mathbf x_1))\prod_{t=2}^{n} p (\mathbf y(\mathbf x_t) \mid \mathbf y(\mathbf x_{t-1}), \mathbf A, \mathbf \Sigma_\varepsilon) \biggl\} \\
    & \propto -\frac{n}{2}\log(|\mathbf \Sigma_\varepsilon|) -\frac{\mathbf y(\mathbf x_1)^T \mathbf \Sigma^{-1}_\varepsilon \mathbf y(\mathbf x_1)}{2} \\
    &\qquad - \sum_{t=2}^{n}\frac{ (\mathbf y(\mathbf x_t) - \mathbf A \mathbf y(\mathbf x_{t-1}))^T \mathbf \Sigma^{-1}_\varepsilon (\mathbf y(\mathbf x_t) - \mathbf A \mathbf y(\mathbf x_{t-1}))}{2}.
\end{align*}
}

{Taking the derivative of log-likelihood with respect to $\mathbf A$ and $\bm \Sigma_\varepsilon$,  we obtain the maximum likelihood estimator of $\mathbf A$ and $\bm \Sigma_\varepsilon$:
\begin{align*}
    \hat{\mathbf A} &= \left( \sum_{t=2}^n \mathbf y(\mathbf x_t) \mathbf y(\mathbf x_{t-1})^T \right) \left(\sum_{t=2}^n \mathbf y(\mathbf x_{t-1}) \mathbf y(\mathbf x_{t-1})^T \right)^{+} \\
    & = \mathbf Y_{2:n} \mathbf Y^T_{1:n-1} (\mathbf Y_{1:n-1} \mathbf Y_{1:n-1}^T)^+ 
    =\mathbf{Y}_{2:n} (\mathbf{Y}_{1:{n-1}})^+, \\
    {\hat{\bm \Sigma}_{\varepsilon}} &= \frac{\mathbf y(\mathbf x_1) \mathbf y(\mathbf x_1)^T + \sum_{t=2}^n (\mathbf y(\mathbf x_{t}) - \hat{\mathbf A} \mathbf y(\mathbf x_{t-1}))(\mathbf y(\mathbf x_{t}) - \hat{\mathbf A} \mathbf y(\mathbf x_{t-1}))^T }{n}.
\end{align*}
The last equality in the equation of $\hat{\mathbf A}$ can be derived by performing the singular value decomposition (SVD) to $\mathbf Y_{1:n-1}$ and connecting the SVD with Moore–Penrose pseudo-inverse. 
} 
 \end{proof}

We notice that the maximum likelihood estimator of $\mathbf A$ is equivalent to the solution of DMD shown in Eq.~(\ref{equ:objective_function_DMD}). 
 Here $\mathbf{Y}_{1:(n-1)}^+$ can be computed by  the SVD of $\mathbf{Y}_{1:(n-1)} = \mathbf{U D V}^*$ as $\mathbf{Y}_{1:(n-1)}^+=\mathbf{V} \mathbf{D}^{-1} \mathbf{U}^*$, where $\mathbf{U}^*$ and $\mathbf{V}^*$ denote the conjugate transpose of $\mathbf U$ and $\mathbf V$, respectively, and $\mathbf{D} \in 
 \mathbb{R}^{m \times (n-1)}$ is a rectangular diagonal matrix of non-negative singular values. 
In practice, 
one can keep the  $r$  largest singular values for approximation:  $\mathbf{Y}_{1:(n-1)} \approx \mathbf{U}_r \mathbf{D}_r \mathbf{V}_r^*$, where $\mathbf U_r$ is the first $r$ columns of $\mathbf{U}$, $\mathbf{D}_r$ is a $r \times r$ diagonal matrix containing the first $r$ largest singular values, with $r\leq \text{min}(m,n-1)$, and $\mathbf{V}_r$ is the first $r$ columns of $\mathbf{V}$. 
 Consequently, $\hat{\mathbf{A}}$ can be approximated by
\begin{align}
\label{equ:dmd_Ahat_low_rank}
    \mathbf{\hat{A}} \approx \mathbf{Y}_{2:n} \mathbf{V}_r \mathbf{D}_r^{-1} \mathbf{U}_r^*.
\end{align}
As some singular values may be small, the approximation from the right-hand side of Eq.~(\ref{equ:dmd_Ahat_low_rank}) is typically more stable as it avoids numerical error in computing the diagonal terms in  $\mathbf{D}^{-1}$.

A primary goal of DMD is to identify the nonzero eigenvalues and their corresponding eigenvectors of $\mathbf{A}$,  denoted as $\{\lambda_i, \bm \phi_i \}^{r}_{i=1}$, which can approximate the Koopman eigenvalues and modes, respectively \cite{Brunton_koopman_siam_review_2022}. 
However, directly computing the eigenvalues and eigenvectors of a $m \times m$ matrix $\hat{\mathbf A}$ in Eq.~(\ref{equ:dmd_Ahat_low_rank}) can be costly when $m$ is large. 
To reduce the computational cost, we may project $\mathbf{\hat A}$ onto the column space of $\mathbf{U}_r$ and define $\tilde{\mathbf{A}}$ as
\begin{align}
    \tilde{\mathbf{A}} = \mathbf{U}_r^* \hat{\mathbf{A}} \mathbf{U}_r \approx \mathbf{U}_r^* \mathbf{Y}_{2:n} \mathbf{V}_r \mathbf{D}_r^{-1} \mathbf{U}_r^* \mathbf{U}_r = \mathbf{U}_r^* \mathbf{Y}_{2:n} \mathbf{V}_r \mathbf{D}_r^{-1}.
\end{align}
Denote $\{\lambda_i, \bm{\omega}_i\}^{r}_{i=1}$ to be the eigenpairs of $\tilde{\mathbf A}$ such that $\lambda_i \tilde{\mathbf{A}} = \tilde{\mathbf A}\bm{\omega}_i$.  In \cite{tu2014dmd}, the authors show that  $\lambda_i$ is the DMD eigenvalue, and the corresponding eigenvector of $\mathbf A$, also known as the DMD mode, can be calculated below
\begin{align}
    \bm{\phi}_i = \frac{1}{\lambda_i}\mathbf{Y}_{2:n}\mathbf{V}_r\mathbf{D}_r^{-1}\bm{\omega}_i,
\end{align}
for $i=1,...,r$.

The snapshots at any $t$ can be approximated by DMD modes and eigenvalues with a smaller dimension. Denote $\bm{\Phi}=[\bm \phi_1,\dots,\bm \phi_r]$ and $\bm{\Lambda} = \text{diag}(\lambda_1,\dots,\lambda_r)$, for any $t\geq 1$,  the reconstructed snapshots $\mathbf{\hat y}(\mathbf{x}_{t})$ can be represented as
\begin{align}
    \mathbf{\hat y}(\mathbf{x}_{t}) = \mathbf{\hat A}^{t-1}  \mathbf{y}(\mathbf{x}_{1}) =\bm{\Phi} \bm{\Lambda}^{t-1}\bm{b},
    \label{equ:reconstructed_y}
\end{align}
where $\bm{b} = [b_1,\dots,b_r]^T= \bm{\Phi}^+\mathbf{y}(\mathbf{x}_1)$ represents the mode amplitudes with $\bm \Phi^{+}$ being the  pseudo-inverse of $\bm \Phi$. As $\mathbf{y}(\mathbf x_1)$ may contain measurement error, an alternative way  is to minimize the squared error loss below:  
\begin{equation}
\hat{\bm{b}} = \argmin_{\bm b} \sum_{t=1}^n \lVert \bm{\Phi}\mathbf{\Lambda}^{t-1} \bm{b}-\mathbf{y}(\mathbf{x}_t) \rVert^2,
\end{equation}
where $\lVert \cdot \rVert$ is the $L_2$ norm or Frobenius norm.

Eq.~(\ref{equ:reconstructed_y}) can be applied to any $t$, including those $t^*>n$ with $n$ being the number of observed time points, and thus it can be used for forecasts. When the  observations are noise-free, a more straightforward way is to let $ \mathbf{\hat y}(\mathbf x_n)=\mathbf y(\mathbf x_n)$ and forecast output vector on any $t^*>n$ by 
\begin{equation}
\mathbf{\hat y}(\mathbf x_{t^*})= \mathbf{\hat A}^{t^*-n} \mathbf{y}(\mathbf x_n).
\label{equ:direct_forecast}
\end{equation}

From the  DMD-induced process in Eq.~(\ref{equ:dmd_process}), we have the following lemma, which gives the posterior distribution for forecast.  
\begin{lemma}
\label{lemma:posterior_y_t_star}
Conditional on the observations  $\mathbf Y$ with plug-in estimators of  $\mathbf{\hat A}$ and {$\hat{\bm \Sigma}_{\varepsilon}$}, the posterior distribution of the output vector of DMD-induced process  in Eq.~(\ref{equ:dmd_process}) at any $\mathbf x_{t^*}$ follows a multivariate normal distribution
\begin{equation}
\left(\mathbf{ y}(\mathbf x_{t^*})\mid \mathbf Y,  \mathbf{\hat A}, 
{\hat{\bm \Sigma}_{\varepsilon}}\right) \sim  \mathcal{MN}\left(\mathbf {\hat  y}(\mathbf x_{t^*}), \, \sum^{t^*-n-1}_{i=0} \mathbf{\hat A}^i 
{\hat{\bm \Sigma}_{\varepsilon}} (\mathbf{\hat A}^T)^i   \right), 
\label{equ:dmd_forecast_dist}
\end{equation}
where $\mathbf {\hat  y}(\mathbf x_{t^*})$ follows Eq.~(\ref{equ:direct_forecast}) for any $t^*>n$. 
\end{lemma}

\begin{proof}
We prove this by induction. For $t^*=n+1$, we have 
\begin{align*}
\mathbb E[{\mathbf y}(\mathbf x_{t^*}) \mid  \mathbf Y,  \mathbf{\hat A}, 
{\hat{\bm \Sigma}_{\varepsilon}}
]&=\mathbf{\hat A} \mathbf y(\mathbf x_n), \\
\mathbb V[{\mathbf y}(\mathbf x_{t^*}) \mid  \mathbf Y,  \mathbf{\hat A}, 
{\hat{\bm \Sigma}_{\varepsilon}}] &= 
{\hat{\bm \Sigma}_{\varepsilon}}. 
\end{align*}
Assume for  $t^*>n+1$, we have 
\begin{align*}
\mathbb E[{\mathbf y}(\mathbf x_{t^*}) \mid  \mathbf Y,  \mathbf{\hat A}, 
{\hat{\bm \Sigma}_{\varepsilon}}]&=\mathbf{\hat A}^{t^*-n}\mathbf{y}(\mathbf x_n), \\
\mathbb V[{\mathbf y}(\mathbf x_{t^*}) \mid  \mathbf Y,  \mathbf{\hat A}, 
{\hat{\bm \Sigma}_{\varepsilon}}] &=  \sum^{t^*-n-1}_{i=0} \mathbf{\hat A}^i 
{\hat{\bm \Sigma}_{\varepsilon}} (\mathbf{\hat A}^T)^i.  
\end{align*}

For any $t^*+1$, by the sampling model in Eq.~(\ref{equ:dmd_process}) and the law of total expectation, the posterior mean follows:   
\begin{align*}
\mathbb E[{\mathbf y}(\mathbf x_{t^*+1}) \mid  \mathbf Y,  \mathbf{\hat A}, {\hat{\bm \Sigma}_{\varepsilon}}] =&\mathbb E[\mathbb E[{\mathbf y}(\mathbf x_{t^*+1}) \mid  \mathbf Y,  {\mathbf y}(\mathbf x_{t^*}), \mathbf{\hat A}, {\hat{\bm \Sigma}_{\varepsilon}}]] \\
=& \mathbb E[\mathbf{\hat A} {\mathbf y}(\mathbf x_{t^*}) \mid \mathbf Y, \mathbf{\hat A}, {\hat{\bm \Sigma}_{\varepsilon}}] \\
=& \mathbf{\hat A} \mathbf{\hat y}(\mathbf x_{t^*}) =\mathbf{\hat A}^{t^*+1-n} \mathbf{y}(\mathbf x_{n}). 
\end{align*}
By the sampling model in Eq.~(\ref{equ:dmd_process}) and the law of total covariance,  for any   $t^*+1$, the posterior covariance follows
\begin{align*}
&\mathbb V[{\mathbf y}(\mathbf x_{t^*+1}) \mid \mathbf Y, \mathbf{\hat A}, {\hat{\bm \Sigma}_{\varepsilon}}]\\
=&\mathbb V[\mathbb E[{\mathbf y}(\mathbf x_{t^*+1}) \mid \mathbf Y, \mathbf y(\mathbf x_{t^*}), \mathbf{\hat A}, {\hat{\bm \Sigma}_{\varepsilon}}]]+\mathbb E[\mathbb V[{\mathbf y}(\mathbf x_{t^*+1}) \mid \mathbf Y, \mathbf y(\mathbf x_{t^*}), \mathbf{\hat A}, {\hat{\bm \Sigma}_{\varepsilon}}]]  \\
=&  \mathbb V[ \mathbf{\hat A} {\mathbf y}(\mathbf x_{t^*})\mid \mathbf Y,  \mathbf{\hat A}, {\hat{\bm \Sigma}_{\varepsilon}}] +{\hat{\bm \Sigma}_{\varepsilon}} \\
=&\mathbf{\hat A} \left(\sum^{t^*-n-1}_{i=0} \mathbf{\hat A}^i {\hat{\bm \Sigma}_{\varepsilon}} (\mathbf{\hat A}^T)^i \right) \mathbf{\hat A}^T +{\hat{\bm \Sigma}_{\varepsilon}}
= \sum^{t^*-n}_{i=0} \mathbf{\hat A}^i {\hat{\bm \Sigma}_{\varepsilon}} (\mathbf{\hat A}^T)^i. 
\end{align*}

\end{proof}

Although we can assess the uncertainty of the forecast by DMD from the predictive distribution in Eq.~(\ref{equ:dmd_forecast_dist}), 
{the linear state space model can be restrictive to approximate nonlinear dynamical systems. Furthermore, the uncertainty in estimating $\mathbf A$ and $\bm \Sigma_{\varepsilon}$ is not propagated for predictions}. Finally, choosing the rank $r$ in DMD is an open problem as it represents one's belief on the degree of the model is misspecified, which could be hard to be quantified precisely. Typical ways of choosing $r$ include letting the summation of DMD eigenvalues explain a large proportion of the output variability, while this choice could potentially misfit the data, as minimizing the $L_2$ loss in Eq.~(\ref{equ:objective_function_DMD}) cannot avoid overfitting the data.

\subsection{Higher order dynamic mode decomposition}
One limitation of the DMD approach is that only the observation from the prior time point is used,  equivalently inducing a first-order Markov model in Eq.~(\ref{equ:dmd_process}). Variants of DMD approaches, such as Higher Order Dynamic Mode Decomposition (HODMD) \cite{le2017higher} or Hankel DMD \cite{arbabi2017ergodic}, use more observations from longer time lag to construct the dynamics: 
\begin{align}
\label{equ:HODMD_assumption}
    \mathbf{y}(\mathbf{x}_{t+q}) = \mathbf{A}_1\mathbf{y}(\mathbf{x}_{t}) + \mathbf{A}_2 \mathbf{y}(\mathbf{x}_{t+1}) + \dots + \mathbf{A}_q \mathbf{y}(\mathbf{x}_{t+q-1}),
\end{align}
where $q \geq 1$ is a tunable parameter that determines the number of time-lagged snapshots to be included in the model.

The estimation accuracy from  HODMD  can be higher than the conventional DMD as multiple time-lagged snapshots are used. 
 For scenarios where the number of time points is larger than the number of output coordinates, including more time-lagged snapshots can increase the upper bound of the number of nonzero singular values, thus potentially capturing complex dynamics in a higher dimensional space. 
From the theoretical point of view, the eigenfunctions and eigenvalues of HODMD are guaranteed to converge to the Koopman eigenfunctions and eigenvalues for ergodic systems \cite{arbabi2017ergodic}.

Let us  define  $\mathbf{y}^{\text{aug}}(\mathbf{x}_t)=( \mathbf{y}(\mathbf{x}_t)^T, \mathbf{y}(\mathbf{x}_{t+1})^T,..., \mathbf{y}(\mathbf{x}_{t+q-1})^T)^T $, an augmented vector of  $mq$ dimensions that contain  $q$ snapshots. 
The linear mapping matrix  $\mathbf{\hat{A}}^{\text{HODMD}}$ in HODMD can be obtained by minimizing the Frobenius or $L_2$ norm between the observations and linear dynamics constructed from the previous time steps: 
$   \mathbf{\hat{A}}^{\text{HODMD}}= \argmin_{\mathbf{A}^{\text{HODMD}}}\lVert \mathbf{Y}^{\text{aug}}_{2:n}- \mathbf{A}^{\text{HODMD}}\mathbf{Y}^{\text{aug}}_{1:(n-1)} \rVert$, 
where $\mathbf{Y}^{\text{aug}}_{2:n} = [\mathbf{y}^{\text{aug}}(\mathbf{x}_2),\dots, \mathbf{y}^{\text{aug}}(\mathbf{x}_n)]$ and $\mathbf{Y}^{\text{aug}}_{1:(n-1)} = [\mathbf{y}^{\text{aug}}(\mathbf{x}_1),\dots, \mathbf{y}^{\text{aug}}(\mathbf{x}_{n-1})]$.

However, 
concatenating $q$ consecutive snapshots increases the number of rows in $\mathbf{Y}^{\text{aug}}_{1:(n-1)}$ from $m$ to $mq$, and the cost of a singular value decomposition for $\mathbf{Y}^{\text{aug}}_{1:(n-1)}$, leading to higher computational cost. To overcome this limitation, one can use a subsampled version of the data instead of the entire dataset. For instance, one might skip $\Delta t$ time steps when constructing the data matrices, i.e., $\mathbf{\tilde Y}_{1}^{\text{aug}} = [\mathbf{y}^{\text{aug}}(\mathbf{x}_1), \mathbf{y}^{\text{aug}}(\mathbf{x}_{1+\Delta t}),\dots,\mathbf{y}^{\text{aug}}(\mathbf{x}_{1+\lfloor  \frac{n-2}{\Delta t}\rfloor \Delta t}) ]$ and $\mathbf{\tilde Y}_{2}^{\text{aug}}=[\mathbf{y}^{\text{aug}}(\mathbf{x}_2), \mathbf{y}^{\text{aug}}(\mathbf{x}_{2+\Delta t}),\dots, \mathbf{y}^{\text{aug}}(\mathbf{x}_{2+\lfloor \frac{n-2}{\Delta t}\rfloor \Delta t})]$. The estimator of $\mathbf A^{\text{HODMD}}$ can be computed below 
\begin{align}
\label{equ:objective_func_hodmd_d_s}
    \mathbf{\hat{A}}^{\text{HODMD}}= \argmin_{\mathbf{A}^{\text{HODMD}}}\lVert \mathbf{\tilde Y}^{\text{aug}}_{2}- \mathbf{A}^{\text{HODMD}}\mathbf{\tilde Y}^{\text{aug}}_{1} \rVert. 
\end{align}
The {HODMD} contains two prespecified parameters: the number of time-lagged snapshots  $q$ to be included in any given time, and the number of skipped time steps $\Delta t$ in estimation. 
Note that the estimated $\mathbf A^{\text{HODMD}}$ does not preserve the model structure in Eq.~(\ref{equ:HODMD_assumption}). Instead, let us consider the following {generative} model for  $\text{HODMD}$ with parameters $(q, \Delta t)$ 
\begin{equation}
\mathbf y^{\text{aug}}(\mathbf x_{2+i\Delta t })=\mathbf A^{\text{HODMD}} \mathbf y^{\text{aug}}(\mathbf x_{1+i\Delta t })+\bm \varepsilon^{\text{aug}}_{{2+i\Delta t}}, 
\label{equ:HODMD_sampling_model}
\end{equation}
for $i=1,...,\lfloor (n-2)/\Delta t\rfloor$ with $\bm \varepsilon^{\text{aug}}_{{2+i\Delta t}}\sim \mathcal{MN}(\mathbf 0, \bm \Sigma^{aug}_{\varepsilon})$. {In practice, the model performance depends on the choice of $(q, \Delta t)$, since the {underlying data-generating} model may rely on multiple time-lagged snapshots and using observations from longer time lag make the model more accurate.} The $\text{HODMD}$ estimator  in Eq.~(\ref{equ:objective_func_hodmd_d_s}) is equivalent to the MLE of $\mathbf A^{\text{HODMD}}$ in the {generative} model in Eq.~(\ref{equ:HODMD_sampling_model}).

\subsection{Extended dynamic mode decomposition}
\label{sec:edmd}

The {generative} model of the DMD algorithm in Eq.~(\ref{equ:dmd_process}) is a linear state space model, while some dynamical systems cannot be accurately approximated by linear dynamics.  
The extended dynamic mode decomposition (EDMD) \cite{williams2015data} aims to define a dictionary of {$\tilde{m}$} nonlinear basis functions {$\mathbf k(\cdot)=(k_1(\cdot),\dots,k_{\tilde m}(\cdot))^T$} to lift the observations to a system that can be approximated by linear dynamics.  
Denote the linear mapping matrix $\mathbf{A}^{\text{EDMD}}$ to be an approximation of the Koopman operator. In EDMD, the linear mapping matrix $\mathbf{A}^{\text{EDMD}}$ is obtained by minimizing the squared error loss function
\begin{align}
    \mathbf{\hat{A}}^{\text{EDMD}} = \argmin_{\mathbf{A}^{\text{EDMD}}} \sum_{t=1}^{n-1} \lVert \mathbf{k}(\mathbf{y}(\mathbf x_{t+1})) - \mathbf{A}^{\text{EDMD}} \mathbf{k}(\mathbf{y}(\mathbf x_t))  \rVert^2_2.
\end{align}

Similar to the DMD-induced process, the estimator of EDMD is equivalent to the maximum likelihood estimator of the linear mapping matrix in a linear state space model defined in the lifted space for $t=1,\dots,n-1$: 
\begin{align}
    \mathbf{k}(\mathbf{y}(\mathbf{x}_{t+1})) = \mathbf{A}^{\text{EDMD}}\mathbf{k}(\mathbf{y}(\mathbf{x}_{t})) + \bm \varepsilon^{\text{EDMD}}_{t+1},
    \label{equ:EDMD_loss}
\end{align}
 with $\bm \varepsilon^{\text{EDMD}}_{t+1} \sim \mathcal{MN}(\mathbf 0, {\bm \Sigma_{\varepsilon}^{\text{EDMD}}} 
)$, where $\bm \Sigma^{\text{EDMD}}_{\varepsilon}$ is a positive definite matrix.  

After estimating the linear mapping matrix between the linear state space model, we need to transform it back   to predict the future states \cite{korda2018linear}, which may be achieved by defining 
 $   \mathbf{\hat y}(\mathbf x_{t}) = \mathbf P \mathbf{k}(\mathbf y(\mathbf x_t))$,  
where $\mathbf{P}$ is a $m \times \tilde{m}$ matrix and can be estimated by 
  $  \mathbf{\hat{P}} = \argmin_{\mathbf P}\sum_{t=1}^n \lVert \mathbf{y}(\mathbf x_{t}) - \mathbf P \mathbf{k}(\mathbf y(\mathbf x_t)) \rVert ^2$.

Choosing an appropriate set of basis functions is crucial for the EDMD method. 
A few generic basis functions, such as Hermite polynomials, radial basis functions, and discontinuous spectral elements, were suggested in  \cite{williams2015data}. The selection of basis functions depends on the context of the problem and domain knowledge may be used as well, whereas misspecifying basis functions can degrade the estimation efficiency of the model. 
PP-GP is closely connected to EDMD. Assuming the mean is zero, we can write the predictive mean of PP-GP in a matrix form 
$\mathbf {\hat Y}=\mathbf W\mathbf {K} $, 
where $\mathbf Y$ is a $m\times n$ observational matrix of $n$ snapshots,   
and $\mathbf { W}=[\mathbf { w}^T_1,...,\mathbf {w}^T_m]^T$ is an {$m\times n$} weight matrix given from {Corollary}  \ref{corollary:pred_mean}. This means the prediction from PP-GP uses the same kernel basis to represent the output for each coordinate, whereas the weights are estimated by solving the linear system of equations with shared coefficients, separately for the output at each coordinate. 
Compared to EDMD, the PP-GP does not project the output onto the lifted space, and hence we do not need to transform the lifted states back for forecasting. 
The PP-GP is a flexible model, as the mean and variance parameters are distinct for each coordinate, which can be marginalized out by computing the predictive distribution. Besides, the covariance matrix contains range and nugget parameters, and they are estimated by the maximum marginal posterior distribution, discussed in \hyperref[sec:appendix]{Appendix}.

\subsection{Computational complexity}
{The computational complexity of estimating the transition and covariance matrices in DMD are $\mathcal{O}(\text{min}(m^2n, mn^2))$ and $\mathcal{O}(m^2n)$, respectively. {Obtaining the} $n^*$-step {forecast} with uncertainty quantification by DMD requires $\mathcal{O}(m^2rn^*)$, where $r$ is the rank of the observation matrix $\mathbf Y_{1:n-1}$. Similarly, for HODMD with time lag $q$ and thining parameter $\Delta t$, the computational complexity of parameter estimation and $n^*$-step forecast with uncertainty quantification are $\mathcal{O}(\text{min}((mq)^2\lfloor \frac{n}{\Delta t} \rfloor, mq(\lfloor \frac{n}{\Delta t} \rfloor)^2)) + \mathcal{O}((mq)^2\lfloor \frac{n}{\Delta t} \rfloor)$ and $\mathcal{O}((mq)^2rn^*)$, respectively. For EDMD using $\tilde{m}$ radial basis functions to lift the data where the centers are determined by the K-means clustering approach, transforming the data and estimating the parameters in the lifted space requires $\mathcal{O}(mn\tilde{m}T)$, where $T$ is the number of iterations in K-means. Obtaining the predictions needs $\mathcal{O}(m \tilde{m} n^*)$. For simplicity, here we assume $m>n>\tilde{m}$. The computational complexity for estimating the parameters as well as providing forecast with uncertainty assessment by DMD, HODMD, EDMD, and PP-GP are summarized in Table \ref{tab:complexity}.}

\begin{table}[t]
\centering
\begin{tabular}{lll}
\hline
 & Parameter estimation & 
 Forecast and UQ  \\
\hline 
 DMD & $\mathcal{O}({m^2n})$ &  $\mathcal{O}(m^2rn^*)$\\

HODMD & $\mathcal{O}({(mq)^2 \lfloor\frac{n}{\Delta t}\rfloor})$ & 
$\mathcal{O}({m^2} q^2 r n^*)$ \\
EDMD & $\mathcal{O}(mn\tilde{m}T)$ & $\mathcal{O}(m\tilde{m} n^*)$ \\
PP-GP & $\mathcal{O}(\tilde{S}n^3+\tilde{S}n^2m)$ & $\mathcal{O}(Smn^2n^*+n^3 )$\\ 
\hline
\end{tabular}
\caption{{Computational complexity for parameter estimation and forecast of $n^*$ time points in DMD, HODMD, EDMD, and PP-GP, where $m$ is the dimension of an output, $n$ is the number of observations, $r$ is the rank of data matrix, $q$ is the number of snapshots stacked in HODMD, $\Delta t$ is the number of skipped time points in HODMD, $\tilde{m}$ is the number of basis functions in EDMD, $T$ is the number of iterations in K-means, $\tilde{S}$ and $S$ are the iterations of numerical optimization and samples for forecast in PP-GP. }}
\label{tab:complexity}
\end{table}

\section{Connection of different data-driven approaches of modeling dynamical systems with respect to  {generative models}}
\label{sec:connection}
Here we compare three classes of data-driven models, namely the proper orthogonal decomposition (POD) \cite{berkooz1993proper}, DMD and PP-GP approaches. To simplify the notations, we assume the data are properly centered, meaning that the $m \times n$ real-valued output matrix $\mathbf Y$ has zero mean. 

In POD, the data are decomposed by SVD $\mathbf Y=\mathbf U_y \mathbf D_y \mathbf V_y^T$, where $\mathbf U_y$ and $\mathbf V_y$ are $m\times m$ and $n\times n$ unitary matrix, respectively, and  $\mathbf D_y$ is a $m \times n$ rectangle diagonal matrix with non-negative singular values in the diagonals. The first $r \leq m$ columns of $\mathbf U_y$ associated with the largest $r$ singular values provide the orthogonal basis of a linear subspace to reconstruct the covariance of output at different coordinates by treating the temporal observations as independent measurements:  $\mathbf Y \mathbf Y^T/(n-1)=\mathbf U_y \mathbf D^2_y \mathbf U^T_y/(n-1)$. This approach is known as the principal component analysis (PCA) \cite{jolliffe2002principal}, which is widely used in unsupervised learning and dimension reduction. The SVD basis from the  PCA can be shown to have the same linear subspace to the maximum marginal likelihood estimator  of {the linear mapping matrix} $\mathbf B$ after marginalizing out $\mathbf z(\mathbf x_t)$ in the following {generative} model  \cite{tipping1999probabilistic}:  
\begin{equation}
\mathbf y(\mathbf x_t)=\mathbf B \mathbf z(\mathbf x_t) +\bm \epsilon_t, 
\label{equ:ppca}
\end{equation}
where $\bm \epsilon_t \sim \mathcal{MN}(\mathbf 0, \sigma^2_0 \mathbf I_m)$ and $\mathbf z(\mathbf x_t) \sim  \mathcal{MN}(\mathbf 0, \mathbf I_r)$ is a $r$-dimensional latent factors independently following  standard normal distributions. 
Under such model, the covariance of the data at each time follows $\mathbb V[\mathbf y(\mathbf x_t)]=\mathbf B \mathbf B^T+\sigma^2_0\mathbf I_m$. 
However, the {generative} model assumes independence between the observations at different time points, which is restrictive. In \cite{gu2020generalized},  $z_l(\cdot)$ is modeled as a GP for each $l=1,...,r$, and the  maximum marginal likelihood estimator of $\mathbf B$ under the assumption $\mathbf B^T\mathbf B=\mathbf I_r$ is derived. 

Second, the {generative} model of DMD is given in Eq.~(\ref{equ:dmd_process}), where the noise of the data is not modeled. Assuming the initial states follow a multivariate normal distribution with zero mean and covariance {$\bm \Sigma_{\varepsilon}$}. It is not hard to show that  $\mathbb E[\mathbf y(\mathbf x_t)]=\mathbf 0$ and the covariance between any two output vectors at two time points follows: 
$\mbox{Cov}[\mathbf y(\mathbf x_{t'}),\mathbf y(\mathbf x_{t})]=\mathbf A^{t'-t}\sum^{t-1}_{i=0} \mathbf A^i {\bm \Sigma_{\varepsilon}} (\mathbf A^{T})^i$ for $t'\geq t$. Specifically, the covariance of output at any time point $t\geq 1$ follows  $\mathbb{V}[\mathbf y(\mathbf x_{t})]=\sum^{t-1}_{i=0} \mathbf A^i {\bm \Sigma_{\varepsilon}} (\mathbf A^{T})^i $. Compared with the sampling model in POD, the {generative} model in the DMD-induced process in Eq.~(\ref{equ:dmd_process}) are correlated over time and the strength of the correlation is captured by the linear mapping matrix $\mathbf A$ between output vectors at the consecutive time points. 
The DMD-induced process may not be differentiable with respect to time and the assumption of homogeneous variance at each output coordinate may also be restrictive for applications where the output has different scales. 

Third, the  PP-GP model in Eq.~(\ref{equ:hat_y_j}) 
has the same predictive mean as modeling the output matrix $\mathbf Y$ by a matrix-normal distribution 
 \cite{gu2016parallel},  with a separable covariance $\mathbb V[\mathbf Y]=\bm \Sigma_y \otimes \mathbf {\tilde K}$, where $\bm \Sigma_y$ is the covariance between output coordinates with the $j$th diagonal term being $\sigma^2_j$, for $j=1,...,m$ and  $\mathbf {\tilde K}$ is the correlation matrix between inputs with $\otimes$ denoting the Kronecker product. 
In comparison, the {generative} model by DMD in Eq.~(\ref{equ:dmd_process}) is a linear state space model, which has a semi-separable covariance structure.   The PP-GP induces nonlinear dynamics when using the observations from previous time points as the inputs, and the differentiability of the nonlinear processes induced by PP-GP can be controlled by the choice of kernel function. 
When the underlying dynamic is smooth, {a} differentiable {GP} prior of the nonlinear dynamics by PP-GP may be preferred to have a better {convergence} rate of compared to a {nondifferential} GP prior \cite{van2008rates}. 
Another advantage of PP-GP is that the range parameters can be estimated by the MLE or maximum marginal posterior mode, which is more flexible than using fixed nonlinear basis functions in EDMD. 
Lastly, the variance of the output coordinate is distinct and the variance estimator of PP-GP has a closed form expression in Eq.~(\ref{equ:sigma_j_2}), whereas the induced processes by DMD and its variants typically have homogeneous variance. The different variance terms make PP-GP particularly suitable when the output has different scales, which are common in practice.

\section{Numerical results}
We compare different data-driven forecast approaches for nonlinear dynamical systems, focusing on uncertainty quantification of the forecast. We consider two scenarios. In the first scenario, we assume the underlying dynamical system is modeled by a map from  $\mathbb R^{p} \to \mathbb R^m$: $d \mathbf y/dt=\mathbf f(\mathbf x_t)$, where the input variables $\mathbf x_t$ is a subset of $\mathbf y_{t}$. Here the vector-valued function $\mathbf f$ is treated as unknown and required to be approximated, whereas the inputs $\mathbf x_t$ are known. For all approaches, we do not include the vector-valued function $\mathbf f(\cdot)$  in nonlinear basis functions. Instead, we test uncertainty quantification with generic kernels or nonlinear basis functions that provide default ways of approximation. 
In the second scenario, the {dynamical system} is described by $d \mathbf y/dt=\mathbf f(\mathbf x_t, \mathbf u_t)$, where we can only observe $\mathbf x_t$, whereas the external inputs $\mathbf u_t$ and the vector-valued function $\mathbf f(\cdot)$ are unobserved. 

For both scenarios, we forecast held-out data $\mathbf y(\mathbf x_{t^*})=(y_1(\mathbf x_{t^*}),...,y_m(\mathbf x_{t^*}))^T$ at $t^*=n+1,n+2,...,n+n^*$. 
 {The autoregressive model with lag order 1 (AR(1)) separately fitted for each coordinate is used as a benchmark prediction model \cite{petris2009dynamic}}. Also included are   DMD,  HODMD and PP-GP.  {Note that the {generative} model of the DMD method is equivalent to a noise-free vector autoregressive (VAR) model with lag order 1 \cite{prado2010time} and hence we do not include other VAR models for comparison.} 
The criteria include the predictive root of mean squared error (RMSE), the average length of the $95\%$ predictive intervals ($L(95\%)$), and the proportion of the samples covered in the $95\%$ predictive interval ($P(95\%)$): 
\begin{align}
\mbox{RMSE}&=\left({{\sum^{m}_{j=1}\sum^{n+n^*}_{t=n+1} (\hat y_j(\mathbf x_{t^*})-  y_j(\mathbf x_{t^*}))^2}}\right)^{\frac{1}{2}},\label{equ:nrmse}\\
L(95\%)&=\frac{1}{mn^*}\sum^{m}_{j=1} \sum^{n^*+n}_{t^*=n+1}\mbox{length}\left\{ CI_{j,t^*}(95\%) \right\}, \label{equ:length}\\
P(95\%)&=\frac{1}{mn^*}\sum^{m}_{j=1} \sum^{n^*+n}_{t^*=n+1}1_{y_j(\mathbf x_{t^*}) \in CI_{j,t^*}(95\%)} \label{equ:coverage},
\end{align}
where $\hat y_j(\mathbf x_{t^*})$ is the prediction of the output at coordinate $j$ with input $\mathbf x_{t^*}$, $CI_{j,t^*}(95\%)$ is the $95\%$ predictive interval of the output at coordinate $j$ and time $t^*$,  $\mbox{length}\left\{ CI_{j,t^*}(95\%)\right\}$ denotes the length of the predictive interval, and $\bar y=\sum^{m}_{j=1}\sum^{n}_{t=1} y_j(\mathbf x_{t})/{(mn)}$ is the mean of the observations in the training data set. An accurate method should have small predictive error quantified by RMSE, short average length of $95\%$ predictive interval ($L(95\%)$) and the proportion of the sample covered by the $95\%$ predictive interval ($P(95\%)$) should be close to the $95\%$ nominal level.

\label{sec:numeric}
\subsection{Lorenz 96 system}
We first discuss the Lorenz 96 system for modeling the atmospheric quantities at equally spaced locations along a cycle \cite{lorenz1996predictability}: 
\begin{equation}
\frac{d y_j(t)}{dt}=(y_{j+1}(t)-y_{j-2}(t))y_{j-1}(t)-y_{j}(t)+F,
\end{equation}
for $j=1,...,m$, where $m=40$ and $F=8$ are typically used for testing. Here $f_j(\mathbf x_t)=(y_{j+1}(t)-y_{j-2}(t))y_{j-1}(t)-y_{j}(t)$, where the 4 dimensional input is $\mathbf x_t=\{y_{j-2}(t), y_{j-1}(t),y_j(t),y_{j+1}(t)\}$. 
The Lorenz 96 system is often used for demonstrating the effectiveness of nonlinear filtering approaches, such as ensemble Kalman filter in data assimilation \cite{evensen2009data},  where the function $f_j(\cdot)$ is typically assumed to be known. 
Here we assume the underlying dynamics from $f_j(\cdot)$ is unknown. 

We test a few methods and compare their performance on uncertainty quantification. We assume both the derivative and output values are available. The data are obtained by the Runge Kutta method of order 4 with step size $h=0.01$ for $1000$ steps. The initial values of the states are sampled from zero mean multivariate normal distribution where covariance matrix is sampled from a Wishart distribution with the scale matrix being identity and $m$ degrees of freedom \cite{roth2017ensemble}. Any method can use the  $mn=4,000$ observations from the first $n=100$ time points as training observations, whereas the rest of the $36,000$ observations at later $n^*=900$ time points are held out as test data. For DMD and HODMD, we try both the observed output values and derivatives for forecasting. Since both ways do not work well, we only present results based on the observed output values. When constructing the data matrices for HODMD, 6 observed snapshots ($q=6$) are concatenated and 3 time points are skipped in the augmented data ($\Delta t =3$). For PP-GP, we uniformly subsample $n_{training}=500$ observations from $4,000$ observations of derivatives in the training time period to estimate the parameters and construct predictive distributions in Eq.~(\ref{equ:pred_distr}), because of the high computational cost when $n$ is large. We use the default product Mat{\'e}rn covariance function with roughness parameter being 2.5 in PP-GP. 
As PP-GP can be considered as an extended version of DMD on projecting the data onto the kernel space discussed in Section \ref{sec:edmd}, we do not include any other EDMD approach. The range parameters of the kernel in PP-GP are estimated by the default marginal posterior mode estimation \cite{gu2018robustgasp}, which is more flexible than assuming fixed nonlinear basis function in EDMD. 

\begin{figure}[t]
\centering
\includegraphics[scale=0.52]{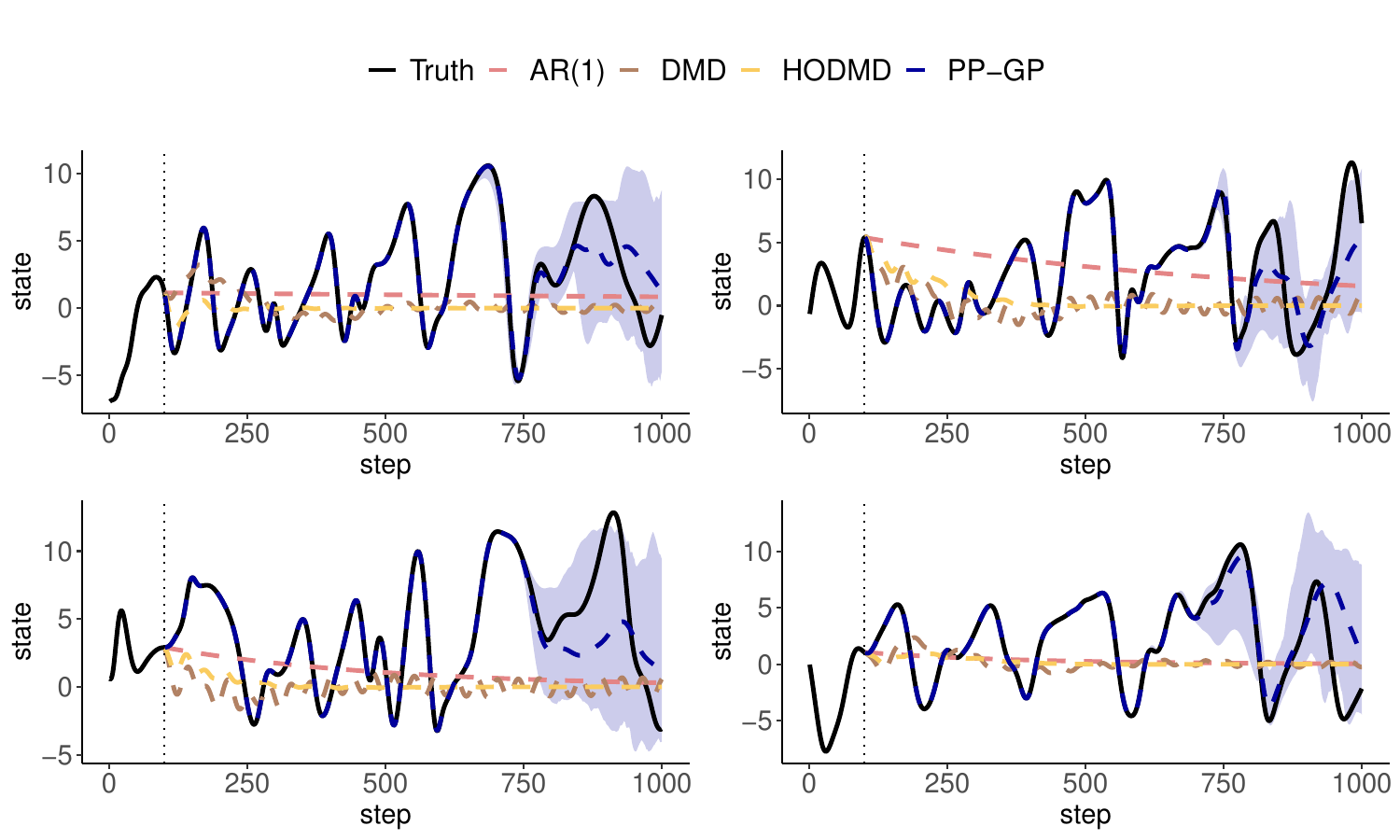}
\caption{ 
Forecast of the Lorenz 96 systems for 900 steps by {AR(1) (pink dashed curves)}, DMD (brown dashed curves), HODMD (yellow dashed curves) and PP-GP (blue dashed curve). The $95\%$ predictive interval by PP-GP is graphed as blue shaded area. The blue dashed curves (PP-GP) and black curves (held-out truth) overlap for around the first 500 held-out time steps.}
\label{fig:eg_1_900_pred}
\end{figure}

\begin{figure}[t]
\centering
\includegraphics[scale=0.65]{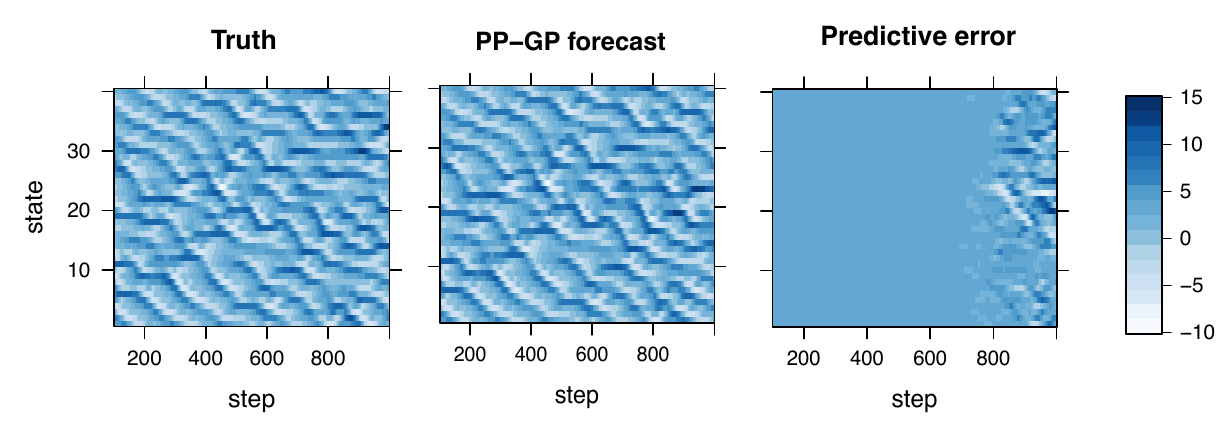}
\caption{The truth, 900-step forecast by PP-GP, and their difference for the Lorenz 96 system. }
\label{fig:eg1_true_ppgp_diff}
\end{figure}

\begin{table}
\centering
\begin{tabular}{llll}
\hline
 500-step forecast &RMSE & $P(95\%)$ & $L(95\%)$ \\ 
\hline
{AR(1)} &{4.60} & {69.8\%} & {10.8}\\
DMD   & 4.51 & {91.6\%} & {20.2}\\ 
HODMD & 4.24 & {98.8\%} & {39.4}  \\  
PP-GP & $0.0126$ & $93.9\%$ & $0.0352$ \\ 
 \hline
900-step forecast &RMSE & $P(95\%)$ & $L(95\%)$ \\ 
\hline
{AR(1)} & {4.63} & {75.4\%} & {12.6} \\
DMD  & 4.55  & {93.5\%} & {21.9} \\ 
HODMD & 4.37 & {99.3\%} & {43.6}  \\  
PP-GP & $1.52$ & $94.6\%$& $2.64$\\ 
 \hline
\end{tabular}
\caption{Forecast accuracy and uncertainty assessment on the held-out data. The standard deviation is $3.54$ and $3.62$ for the 500-step test data and 900-step test data, respectively. 
}
\label{tab:eg1_table}
\end{table}

\begin{figure}[t]
\centering
\includegraphics[scale=0.65]{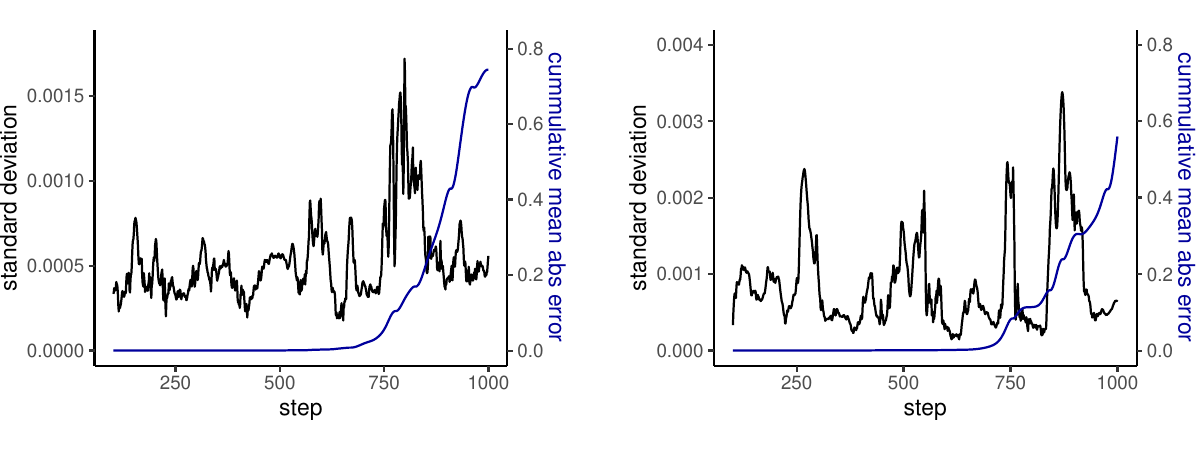}
\caption{The predictive standard deviation from the PP-GP model at each step and cumulative mean absolute error $SE_j(t^*)=\sum^{t^*}_{s^*=n+1}|\hat y_j(s^*)-y_j(s^*)|/(t^*-n)$ of 900-step forecast of state $j=1$ and $j=21$ for $t^*=101,...,1000$ and $n=100$. }
\label{fig:se_cumsum}
\end{figure}

Fig. ~\ref{fig:eg_1_900_pred} gives the 900-step forecast by {AR(1),} DMD, HODMD and PP-GP for the $10$th $20$th, $30$th and $40$th states. The uncertainty of the forecast by PP-GP is graphed as the blue shared area in all plots. 
With the default kernel function and estimation \cite{gu2018robustgasp}, the forecast of PP-GP is reasonably accurate for the first 500 time steps and $95\%$ predictive interval by PP-GP (graphed as the blued shaded area) is almost indistinguishable {in this time domain}. {As the average Lyapunov time for our simulation is around 1.58, the PP-GP can make precise forecast for more than 3 Lyapunov time}. 
The $95\%$ predictive interval becomes noticeably wider at later time steps, and simultaneously, the difference between PP-GP and held-out truth becomes large. The internal uncertainty assessment {quantified by the 95\% predictive interval of} PP-GP  provides a time range of reliable forecasts without knowing the held-out truth. Furthermore, Fig.~\ref{fig:eg1_true_ppgp_diff} compares the truth to the forecast by  PP-GP for all states, which confirms the forecast by PP-GP for the first 500 steps is accurate for all states.

Table \ref{tab:eg1_table} summarizes the performance of forecast and uncertainty assessment by different approaches for the Lorenz 96 system. The  RMSE by the PP-GP for the first 500 steps is much smaller than the standard deviation of the test data and the length of the $95\%$ predictive interval by PP-GP  is also substantially smaller than the variability in the held-out observations. Even if the $95\%$ predictive interval by PP-GP is short, it covers $93.9\%$ of the observations, indicating the uncertainty of the forecast is properly quantified. The predictive error by PP-GP becomes large at later steps due to the accumulation of the approximation error, and the overall predictive error of the 900-step forecast is dominated by the large error at later time points. Note that the length of the $95\%$ predictive interval from PP-GP increases automatically in PP-GP, enabling $94.6\%$ of the held-out data to be covered by the $95\%$ predictive interval. In comparison, {although the proportion of the samples covered by the $95\%$ predictive interval by DMD and HODMD is also close to the $95\%$ nominal level, the average length of intervals is substantially {larger}}, as the underlying dynamics cannot be written as a linear combination of previous outputs. It is worth mentioning that the uncertainty of DMD and HODMD is affected by the selected rank to represent the data, here chosen to be the smallest value such that the summation of the eigenvalues explain at least $99\%$ of the variability. A principled way to model the noise and select the rank may improve the uncertainty assessment of these methods.

Fig.~\ref{fig:se_cumsum} presents the predictive standard deviation of PP-GP  and the corresponding cumulative mean absolute error between the true values and PP-GP forecasts, for the 1st and 21st states. In the initial 500-step forecasts, both the standard deviation computed by Eq.~(\ref{equ:pred_distr}) and the cumulative mean absolute error are relatively small. 
As the number of forecast steps increases, the cumulative mean absolute error increases. The 95\% predictive interval in Fig.~\ref{fig:eg_1_900_pred} can be used to quantify the time when the forecast becomes inaccurate. 

The uncertainty assessment by the PP-GP model is reasonably accurate as we convert the challenging problem of forecasting chaotic systems to the problem of predicting the one-step-ahead function in a 4-dimensional input space. 
In practice, reducing the inputs to a low-dimensional space is helpful for producing reliable forecasts and uncertainty assessments. 

\subsection{Time-dependent Green's function}

{In the second example, we apply the data-driven approach to forecast the simulation of the nonequilibrium dynamics of interacting electrons in materials exposed to intense time-varying electric fields, which is one of the most challenging problems in condensed matter physics. Modeling nonequilibrium dynamics from the basic principles of quantum mechanics is exceptionally challenging in computation, and has only been extended to the \textit{ab initio} simulation of real materials in the last decade through the development of a new time-domain approach based on a formalism of Keldysh, Kadanoff and Baym \cite{KB1962,Keldysh1965,Stefanucci2013} for the dynamics of the nonequilibrium Green's function~\cite{Attaccalite2011,chan2021giantexciton,Perfetto2022}. We refer to the new computational approach as the time-dependent adiabatic GW (TD-aGW) method, where G stands for the interacting single-particle Green's function and W stands for the screened Coulomb interaction. Details about our implementation of the method can be found in~\cite{chan2021giantexciton}, which is built on approximations developed in~\cite{Attaccalite2011}. In principal, the TD-aGW approach gives quantitatively accurate descriptions of materials' response to electric fields, allowing for the simulation of experiments with femtosecond resolution and the development of new classes of quantum materials whose properties can be switched with laser pulses~\cite{Perfetto2016,Perfetto2020,Chan2023}. However, the scaling with respect to the size of the problem is computationally prohibitive, and thus, the TD-aGW method is currently limited to systems containing a few atoms and for short timescales.}

{In the context of many-body perturbation theory and second quantization, the nonequilibrium Green's function is a two-point correlation function comprised of the creation and annihilation field operators that increase and decrease the number of particles in a given state respectively. In general, it is a function of two time variables, but following the work in~\cite{Attaccalite2011,chan2021giantexciton}, the change in the energy of an electron due interactions with its environment (i.e. the electron self energy) can be approximated as the equilibrium self energy plus a static nonequilibrium contribution. In this approximation, the key equation that we solve in TD-aGW is an equation of motion for a matrix of single-particle density $\bm \rho(t)$ with index $(m_1m_2,\textbf{k})$
\begin{equation}
    i\hbar \frac{\partial}{\partial t} \rho_{m_1 m_2,\mathbf{k}} (t) = [\mathbf H(t), \bm \rho(t)]_{m_1m_2,\mathbf{k}},
    \label{eq:tdagw}
\end{equation}
where $\mathbf H(t)$ is a matrix of the Hamiltonian of the system having the same size of $\bm \rho(t)$,  the square bracket denotes the commutator of two matrices $\mathbf A$ and $\mathbf B$ such that $[\mathbf A,\mathbf B]=\mathbf A\mathbf B-\mathbf B\mathbf A$, and $\hbar$ is the Planck constant.
The particle density matrix is written in a basis of electron quantum states with a band index $m$ and a wave vector (or crystal momentum) $\textbf{k}$. Hence, the matrix element $\rho_{m_1 m_2,\textbf{k}}(t)$ encodes the entanglement of a state $(m_1,\textbf{k})$ with another state $(m_2,\textbf{k})$. 
 The Hamiltonian of the system is calculated as
\begin{equation}
    \label{eq:ham}
    \mathbf H(t) = \mathbf H_0 - e\mathbf{E}(t) \cdot {\mathbf{r}} + \bm \Sigma^{GW} + \delta\bm \Sigma (t).
\end{equation}
Here, $\mathbf H_0$ is the system's mean-field Hamiltonian at equilibrium; $\bm \Sigma^{GW} $ is the electron self-energy at equilibrium computed within the GW approximation; and $\delta \bm \Sigma$ is the nonequilibrium correction to the electron self-energy. $e\mathbf{E}(t) \cdot {\mathbf{r}}$ describes the coupling of the system with the external electric field, such that $\textbf{E}(t)$ is a spatially-uniform, time-varying electric field; $\textbf{r}$ is the position operator in quantum mechanics; and $e$ is the fundamental charge of the electron. 
$\mathbf H_0$ and $\bm \Sigma^{GW}$ describe equilibrium properties and thus have no time-dependence. $\delta \bm\Sigma (t)$ is a functional of $\bm \rho(t)$:
\begin{equation}
    \label{eq:sigma}
        \delta  \Sigma_{m_1m_2,\mathbf k} (t) = \sum
_{m'_1,m'_2,\mathbf{k'}} \rho_{m_1 m_2,\mathbf{k}-\mathbf{k'}}(t) W_{m_1m_1'm_2m_2',\mathbf{k}-\mathbf{k'}},
\end{equation}
where $W_{m_1m'_1m_2m'_2,\mathbf k-\mathbf{k}'}$ is the equilibrium screened Coulomb interaction computed in the random-phase approximation (RPA); $m_1$, $m'_1$, $m_2$, $m'_2$ are band indices, and $\mathbf{k}$ and $\mathbf{k}'$ are vectors in reciprocal space.
}

We use monolayer MoS\textsubscript{2}--a material where the electron self energy is known to be large~\cite{Qiu2013,Qiu2016,Ugeda2014,Chernikov2014,Wang2018}--as our test system. In order to perform the time evolution in Eq.~(\ref{eq:tdagw}), we first need to compute the equilibrium solution before the electric field is turned on to obtain $\bm \rho(t=0)$, as well as the time-independent matrices $\mathbf H_0$ and $\bm \Sigma^{GW}$ in Eq.~(\ref{eq:tdagw}) and $\mathbf W$ in Eq.~(\ref{eq:sigma}). We do this within the one-shot GW approximation~\cite{Hybertsen1986,Rohlfing2000,Deslippe2012} using the following calculation parameters. 
Our basis includes 4 bands (the top 2 valence bands and the bottom 2 conduction bands) and a uniform grid of $36\times 36\times 1$ \textbf{k}-points in reciprocal space. Density functional theory (DFT) calculations with spin-orbit coupling are performed using the Quantum Espresso package \cite{Giannozzi2009}. We use norm-conserving fully relativistic PBE pseudopotentials from the SG15 ONCV potential library \cite{Scherpelz2016} and a plane wave cutoff energy of 80 Ry. A GW plus Bethe Salpeter equation (BSE) calculation is done as a one-shot calculation on top DFT using the BerkeleyGW package \cite{Deslippe2012}. A dielectric cutoff of 10 Ry and 6000 bands are used in the GW and BSE calculations.  We use the results of the equilibrium calculations to setup Eq.~(\ref{eq:tdagw}) at time t=0, and then we propagate the equation in time with an external electric field that mimics a laser pulse polarized along the $r_1$ direction. The field is ${E}_{r_1}(t)= A\sin\left(\frac{\pi t}{T}\right)^{2}\sin\left(\omega t\right)$, where the constants (in Ryd) $A=0.006$ is the amplitude of the electric field, $\omega = 0.022$ is the frequency of the light, and $T=160$ fs is the duration of the light pulse, all in Rydberg atomic units. ${E}_{r_2}(t)=0$ and ${E}_{r_3}(t)=0$. The electric field parameters are values consistent with typical experimental setups in high harmonic generation (HHG) experiments \cite{liu2017hhg}. The 4\textsuperscript{th} order Runge-Kutta method is then used to update $\rho_{m_1 m_2,\mathbf k}(t)$ and $\delta\Sigma (t)$ at each time step, and the matrix element $\rho_{m_1 m_2,\mathbf k}(t)$ is saved for the single \textbf{k}-point, $\textbf{k}=\textbf{0}$, to be used as the test and training data for our data-driven models.

Our focus in this example is on forecasting the real part of the density matrix from Eq.~(\ref{eq:tdagw}), which in turn is related to the two-time Green's function through the Generalized Kadanoff-Baym ansatz (GKBA)~\cite{Lipavsky1986,Stefanucci2013}. Each snapshot is a 16-dimensional vector:  $\mathbf y_t=\mbox{Vec}(\mathbf \rho_{\mathbf k}(t))$, where $\mathbf \rho_{\mathbf k}(t)$ is a $4\times 4$ matrix with the $(m_1,m_2)$th entry being $ \rho_{m_1 m_2,\mathbf k}(t)$ for $m_1=1,...,4$ and $m_2=1,...,4$, and $\mathbf k=\mathbf 0$. 
 To solve Eq.~(\ref{eq:tdagw}), one needs to compute $\mathbf{E}(t)$, which is a known analytic function, and $\delta \bm \Sigma (t)$, which depends on the density matrix $\rho_{m_1 m_2,\mathbf k-\mathbf k'}(t)$  at other reciprocal lattice vectors $\mathbf k\neq \mathbf 0$. To evaluate the performance of that data-driven approaches,  we only utilize the observations $\mathbf y_t$ to construct the model, which only contains the local information at $\mathbf k=\mathbf 0$, whereas the external field $\mathbf E(t)$ and interactions terms $\delta \bm \Sigma(t)$ from other lesser Green's function $G_{m_1 m_2,\mathbf k-\mathbf k'}^{<}(t)$ are not used. 
 For all methods, we test two scenarios with training time steps being 2500 and 3500, respectively. For {AR(1) and} DMD, all training data are used. For HODMD,  we use the same setting as the previous example with $q=6$ and $\Delta t=3$. 
 For PP-GP, we use an isotropic kernel and 800 pairs of observations, uniformly sampled from the training data for constructing the model, and the output vector of $m=m_1m_2=16$ dimensions from the previous time point is used as the input for predicting the one-step-ahead transition function.

\begin{table}[t]
\centering
\begin{tabular}{llll}
\hline 
1000-step forecast &RMSE & $P(95\%)$ & $L(95\%)$ \\ 
\hline
{AR(1)} & {$6.20 \times 10^{-6}$}  & {23.7\%} & {$3.97 \times 10^{-6}$} \\
DMD   & $2.87 \times 10^{-5}$  & {$8.65\%$}  & {$2.03\times 10^{-6}$} \\ 
HODMD & $1.30 \times 10^{-5}$ & {$9.52\%$}  & {$6.56\times 10^{-7}$} \\  
PP-GP &$3.53 \times 10^{-6}$  & $65.2\%$ & $2.98 \times 10^{-6}$  \\ 
\hline
2000-step forecast &RMSE & $P(95\%)$ & $L(95\%)$ \\ 
\hline
{AR(1)} & {$6.23 \times 10^{-6}$}  & {35.8\%} & {$5.57 \times 10^{-6}$} \\
DMD  & $2.01\times 10^{-3}$  & {$4.35\%$} & {$7.29 \times 10^{-5}$} \\ 
HODMD & $6.22\times 10^{-5}$ & {9.22\%} & {$1.94 \times 10^{-6}$} \\  
PP-GP & $3.42 \times 10^{-6}$ & $75.3\%$ & $3.97 \times 10^{-6}$ \\ 
\hline
\end{tabular}
\caption{Forecast and uncertainty assessment  for the density matrix  using the held-out data. 
Observations at the first 2500 time steps are used to fit the model. The standard deviation is $1.56 \times 10^{-5}$ and  $1.48 \times 10^{-5}$ of the 1000-step test data and  the 2000-step test data, respectively.}
\label{tab:eg2_training_2500}
\end{table}

\begin{table}[t]
\centering
\begin{tabular}{llll}
\hline
1000-step forecast &RMSE & $P(95\%)$ & $L(95\%)$ \\ 
\hline
{AR(1)} & {$1.13  \times 10^{-5}$} & {31.7\%} & {$7.12 \times 10^{-6}$} \\
DMD   & $6.00\times 10^{-6}$ & {$15.9\%$} & {$1.00 \times 10^{-6}$} \\ 
HODMD & $9.29 \times 10^{-6}$ & {$29.6\%$ }& {$1.05 \times 10^{-6}$} \\  
PP-GP & $3.93 \times 10^{-6}$ & $81.7\%$ & $1.89 \times 10^{-5}$\\ 
\hline
2000-step forecast &RMSE & $P(95\%)$ & $L(95\%)$ \\ 
\hline
{AR(1)} & {$1.50 \times 10^{-5}$} & {28.5\%} & {$9.93\times 10^{-6}$} \\
DMD  & $8.83 \times 10^{-6}$ & {$16.3\%$} & {$2.98 \times 10^{-6}$} \\ 
HODMD & $2.24 \times 10^{-5}$ & {$20.2\%$} & {$1.62 \times 10^{-6}$} \\  
PP-GP & $9.13\times 10^{-6}$  & $88.8\%$ & $4.36 \times 10^{-5}$\\ 
\hline
\end{tabular}
\caption{Forecast and uncertainty assessment of the density using the held-out data.  Observations from the first 3500-time steps are used to fit the model. The standard deviation is $1.40 \times 10^{-5}$ and $1.06 \times 10^{-5}$ of the 1000-step test data and the 2000-step test data, respectively.}
\label{tab:eg2_training_3500}
\end{table}

{Table \ref{tab:eg2_training_2500} and  Table \ref{tab:eg2_training_3500} provide the performance of each method when using observations from the first 2500 time steps and first 3500 time step as the training data, respectively.}
The PP-GP has the smallest RMSE for 2000-step forecast among three approaches in both scenarios, smallest RMSE for 1000-step forecast in the first scenario.

 {The misspecification of the input variables, however,  degrades the accuracy of predictions and uncertainty quantification.  The predictive error differs when using two output regimes as the training data, as the trend is different, and neither represents the trend in the held-out test data. The coverage of held-out data by the $95\%$ predictive interval from the PP-GP is the highest among all methods, and having observations from a longer training time period seems to improve the overall coverage of the held-out data. 
}

 \begin{figure}[t]
\centering
\includegraphics[scale=0.8]{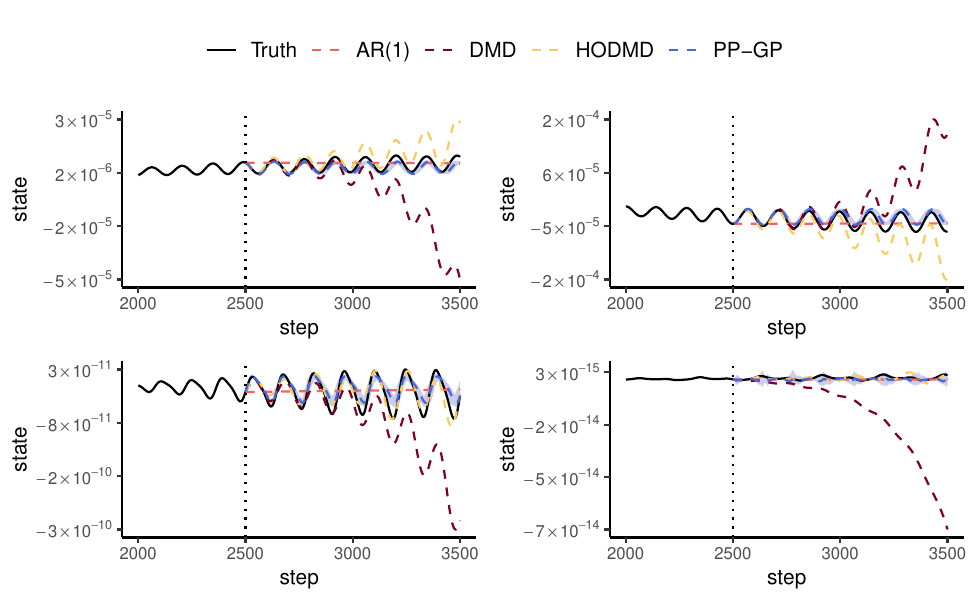}
\caption{{Forecast of two time Green's function for 1000 steps by AR(1) (orange dashed curves), DMD (brown dashed curves), HODMD (yellow dashed curves) and PP-GP (blue dashed curve). 2500-time steps are used as training data. The $95\%$ predictive interval by PP-GP is graphed as blue shaded area.} }
\label{fig:eg_2_pred_n2500}
\end{figure}

\begin{figure}[t]
\centering
\includegraphics[scale=0.8]{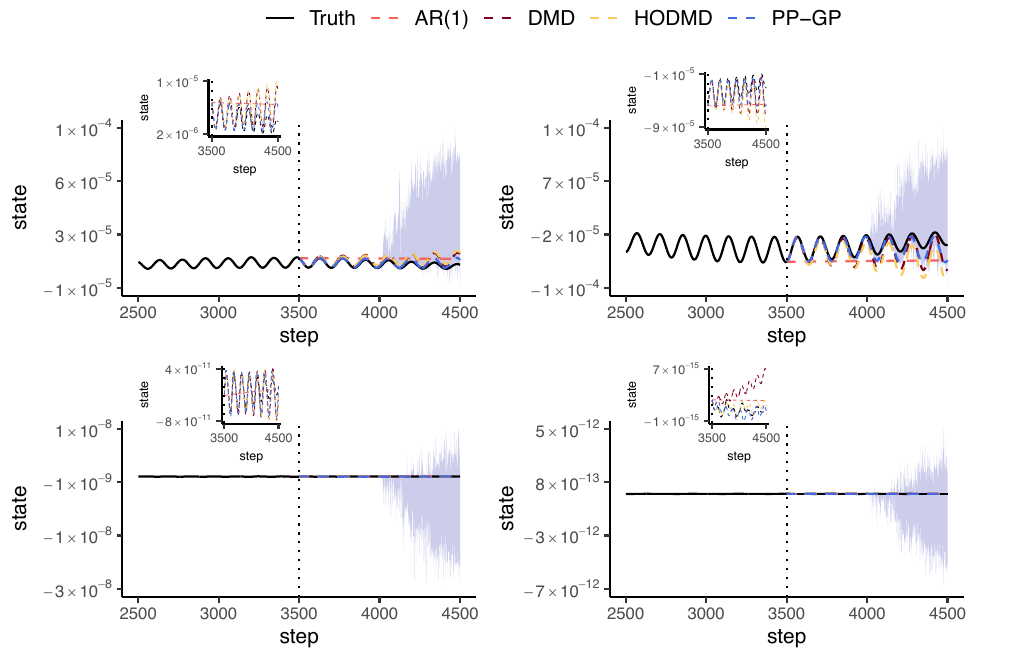}
\caption{Forecast of two time Green's function for 1000 steps by {AR(1) (orange dashed curves)}, DMD (brown dashed curves), HODMD (yellow dashed curves) and PP-GP (blue dashed curve). 3500-time steps are used as training data. The $95\%$ predictive interval by PP-GP is graphed as blue shaded area. }
\label{fig:eg_2_pred_n3500}
\end{figure}

{Fig.~\ref{fig:eg_2_pred_n2500} and Fig.~\ref{fig:eg_2_pred_n3500}} display the 1000-step forecast by {AR(1),} DMD, HODMD and PP-GP model using {2500 and} 3500 timesteps as training data, {,respectively}. 
The PP-GP model can make accurate prediction for the first few cycles with short predictive intervals. The prediction error accumulates, and the model automatically detects the inaccuracy of the prediction, leading to large $95\%$ predictive intervals at later time points. Note that the scale of the held-out truth is decreasing and the overall trend is not captured by any method. This is because for all {four} methods, some inputs such as $\mathbf H(t)$ and $\bm \rho(t)$ at  $\mathbf k\neq \mathbf 0$, are assumed to be unknown. 
The large predictive intervals from PP-GP indicate substantial differences between the output in the training and forecast period, signaling more information is required to obtain an accurate prediction.

\section{Concluding remarks}
\label{sec:conclusion}
  Quantifying the uncertainty for forecast and extrapolation by data-driven models is a challenging task that was not well-studied. We showed popular approaches for representing dynamical systems, such as the dynamic mode decomposition, can be written as the maximum likelihood estimator of a linear mapping matrix in a linear state space model, and this {generative} model allows the uncertainty to be quantified of forecast rigorously. We also extended the parallel partial Gaussian process approach to emulate the one-step-ahead transition function that links observations at two nearby time frames, and propagated the uncertainty through posterior sampling for forecasting a longer time. 
  We numerically compared different approaches with correctly specified inputs and misspecified inputs in two examples.  
  We discussed scenarios where the uncertainty can be reliably quantified, and analyzed the factors that can degrade the accuracy of uncertainty assessment. 

  There is a wide range of open issues to obtain reliable uncertainty quantification for probabilistic forecast of nonlinear dynamical systems. First, restrictive model assumptions, such as equal variance between output coordinates,  subjective choice of latent dimensions and lack of models of trends from the forecast period, can degrade the accuracy of uncertainty assessment for forecasting. Having a {probabilistic generative} model allows one to better understand the model assumptions and hence select data-driven models more suitable for real-world tasks.   Second, the kernel representation of vector functions, such as the PP-GP model, can capture nonlinear behaviors of dynamical systems through modeling the one-step-ahead transition function, whereas the Markov assumption of the model can be restrictive. Having inputs from longer time lag period may improve the model performance. 
  Furthermore, when the dimension of input is large, we need to develop a computationally scalable way to reduce the input dimension and form a suitable distance metric between the reduced inputs. 
 Finally,  filtering approaches may be used along with the data-driven predictions of one-step-ahead transition function, when the observations contain nonnegligible noises.

\section*{Appendix: Estimation of  range and nugget parameters in PP-GP}
\label{sec:appendix} 
The range and nugget parameters of PP-GP model can be estimated from mode estimator, such as the maximum likelihood estimator (MLE) or maximum marginal posterior estimator (MMPE). The MLE can be unstable for estimating these parameters when the sample size is small. Transforming the range parameters to define the inverse range parameter $\beta_l=1/\gamma_l$, we use the MMPE for estimating the inverse range and nugget parameters   (\cite{gu2016parallel}): 
\begin{equation}
(\bm {\hat \beta},\hat \eta)=\mbox{argmax}_{\bm \beta, \eta} \left\{{\log}(\mathcal L(\bm \beta, \eta))+\log(\pi(\bm \beta, \eta))\right\}. 
\label{equ:mmpe}
\end{equation}
Here the logarithm of the marginal likelihood after integrating out the mean and variance is 
\begin{equation}
{\log}(\mathcal L(\bm \beta, \eta))= c_1 -\frac{m}{2}\log|\mathbf {\tilde K}|-\frac{m}{2}\log|\mathbf 1^T_n \mathbf{\tilde K}^{-1} \mathbf 1_n|-\left(\frac{n-1}{2} \right)\sum^{m}_{j=1}\log \left( S^2_j \right),
\end{equation}
where $ S^2_j=(\mathbf y_j-\hat \mu_j\mathbf 1_n)^T \mathbf{\tilde K}^{-1}  (\mathbf y_j-\hat \mu_j\mathbf 1_n)$ and $c_1$ is a normalizing constant not related to $(\bm \gamma, \eta)$. The jointly robust prior \cite{gu2018jointly} is used as a default choice for prior  in the RobustGaSP package \cite{gu2018robustgasp}
\begin{equation}
\log(\pi(\bm \beta, \eta))= c_2+ a \log(\sum^{\tilde p}_{l=1}C_l \beta_{l}+\eta) - b\left(\sum^{\tilde p}_{l=1}C_l \beta_{l}+\eta\right),
\end{equation}
where $c_2$ is a normalizing constant not relevant to $(\bm \beta, \eta)$ and the default choice of prior parameters in the RobustGaSP package is $a=0.2$, $b=n^{-1/\tilde p}(a+\tilde p)$ and $C_l=n^{-1/\tilde p}|x^{max}_l-x^{min}_l|$ with $x^{max}_l$ and $x^{min}_l$ being the largest and lowest input values in the $l$th coordinate, respectively. For deterministic output values, including the cases where numerical error of the simulations is negligible, the nugget $\eta$ may be set to be zero. 

We transform the estimated inverse range parameters back to get $\hat \gamma_l=1/\hat \beta_l$, for $l=1,...,\tilde p$ and   compute  the predictive distribution after integrating the $m$ mean parameters and $m$ variance parameters
\begin{align*}
&p\left(y_j\left(\mathbf x_{t^*}\right) \mid \mathbf X, \mathbf Y, \mathbf x_{t^*}, \bm {\hat \gamma}, \hat \eta\right) \\
&\quad\quad =\int p(y_j\left(\mathbf x_{t^*}\right) \mid \mathbf X, \mathbf Y, \mathbf x_{t^*}, \bm {\hat \gamma}, \hat \eta,\bm \mu, \sigma^2)\pi(\bm \mu, \sigma^2) d\bm \mu d \sigma^2, 
\end{align*}
where  the reference prior of  mean and variance parameters follows  $\pi(\bm \mu, \sigma^2)\propto 1/\prod^{m}_{j=1}\sigma^2_j$,  and $p(y_j\left(\mathbf x_{t^*}\right) \mid \mathbf X, \mathbf Y, \mathbf x_{t^*}, \bm {\hat \gamma}, \hat \eta,\bm \mu, \sigma^2)$ is the conditional distribution of the output at $\mathbf x_{t^*}$ at coordinate $j$. With the assumption the outputs are independent across different coordinates, 
the  predictive distribution $p\left(y_j\left(\mathbf x_{t^*}\right) \mid \mathbf X, \mathbf Y, \mathbf x_{t^*}, \bm {\hat \gamma}, \hat \eta\right)$ follows a Student's distribution in Eq.  (\ref{equ:pred_distr}).

\section*{Acknowledgement}
This research is supported by the National Science Foundation under Award No. 2053423. Yizi Lin acknowledge the support from the support of the UC Multicampus Research Programs and Initiatives (MRPI) program under Grant No.
M23PL5990. Diana Qiu and Victor Chang Lee were supported by the National Science Foundation (NSF) Condensed Matter and Materials Theory (CMMT) program under Grant DMR-2114081. Development of the td-aGW code was supported by Center for Computational Study of Excited-State Phenomena in Energy Materials (C2SEPEM) at the Lawrence Berkeley National Laboratory, funded by the U.S. Department of Energy, Office of Science, Basic Energy Sciences, Materials Sciences and Engineering Division, under Contract No. DE-C02-05CH11231. The calculations used resources of the National Energy Research Scientific Computing (NERSC), a DOE Office of Science User Facility operated under Contract DE-AC02-05CH11231; Anvil at Purdue University through allocation PHY230053 from the Advanced Cyberinfrastructure Coordination Ecosystem: Services \& Support (ACCESS) program, which is supported by National Science Foundation grants \#2138259, \#2138286, \#2138307, \#2137603, and \#2138296~\cite{ACCESS}; and the INCITE program at the Oak Ridge Leadership Computing Facility, which is a DOE Office of Science User Facility supported under Contract DE-AC05-00OR22725. 

 \bibliographystyle{elsarticle-num} 
 \bibliography{References_chronical_2022}

\end{document}